\documentclass[journal,comsoc,final]{IEEEtranTCOM}
\usepackage[T1]{fontenc}

\normalsize

\usepackage{amsmath,amsthm,amssymb}
\usepackage[cmintegrals]{newtxmath}
\usepackage{graphicx,tikz}
\usepackage{bbm}
\usepackage[hidelinks]{hyperref}
\usepackage{amsmath}
\usepackage{dsfont}
\usepackage{circuitikz}
\usepackage{lscape}
\usetikzlibrary{fit}
\usepackage{mathrsfs}
\usetikzlibrary{positioning,babel}
\usepackage[font=small]{caption}
\usepackage{caption}
\usepackage{subcaption}

\usepackage [autostyle, english = american]{csquotes}
\MakeOuterQuote{"}

\newtheorem{thm}{Theorem}
\newtheorem{problem}{Problem}
\newtheorem{lem}{Lemma}

\newtheorem{rem}{Remark}

\newtheorem{property}{Property}

\newtheorem{examp}{Example}

\newtheorem{claim}{Claim}

\usepackage{url}
\usepackage{hyperref}

\DeclareMathOperator{\vecc}{\mathbf{vec}}		
\DeclareMathOperator{\argmin}{arg\,min}
\DeclareMathOperator{\argmax}{arg\,max}

\newcommand{\bc}[1]{\mathbb{B}\textbf{(}#1\textbf{)}}

\let\labelindent\relax

\usepackage{enumerate}
\usepackage[shortlabels]{enumitem}
\newcommand{\mylabel}[2]{#1#2}

\usepackage{fix-cm}
\usepackage{cite}
\usepackage{algorithm}
\usepackage[noend]{algpseudocode}

\usepackage{textcomp}

\newcommand{\blue}[1]{\textcolor{black}{#1}}
\newcommand{\bblue}[1]{\textcolor{blue}{#1}}
\newcommand{\red}[1]{\textcolor{red}{#1}}
\newcommand{\black}[1]{\textcolor{black}{#1}}

\usepackage{eqparbox}

\makeatletter
\def\BState{\State\hskip-\ALG@thistlm}
\makeatother

\makeatletter
\renewcommand{\ALG@beginalgorithmic}{\small}
\makeatother

\begin{document}

\title{An Approximation Algorithm for Optimal Clique Cover Delivery in Coded Caching}


\author{Seyed Mohammad Asghari, Yi Ouyang, Ashutosh Nayyar,~\IEEEmembership{Senior Member,~IEEE,} \\
 and A. Salman Avestimehr,~\IEEEmembership{Senior Member,~IEEE} 
 \thanks{This paper was presented in part at the proceedings of International Conference on Communications (ICC), 2018 \cite{Asghari_ICC_2018}.
 }
\thanks{S. M. Asghari, A. Nayyar, A. S. Avestimehr are with the Department of Electrical
Engineering, University of Southern California, Los Angeles, CA.
Y. Ouyang is currently with Preferred Networks America, Inc., Berkeley, CA, USA. (e-mail: asgharip@usc.edu; ouyangyii@gmail.com; ashutosn@usc.edu; avestimehr@ee.usc.edu).}
}



%

\maketitle
\begin{abstract}
Coded caching can significantly reduce the communication bandwidth requirement for satisfying users' demands by utilizing the multicasting gain among multiple users. Most existing works assume that the users follow the prescriptions for content placement made by the system. However, users may prefer to decide what files to cache. To address this issue, 
we consider a network consisting of a file server connected through a shared link to $K$ users, \black{each equipped with a cache which has been already filled arbitrarily}. Given an arbitrary content placement, the goal is to find a delivery strategy for the server that minimizes the load of the shared link. \black{In this paper, we focus on a specific class of coded multicasting delivery schemes known as the "clique cover delivery scheme"}. We first formulate the optimal \black{clique cover delivery problem} as a combinatorial optimization problem. Using a connection with the weighted set cover problem, we propose an approximation algorithm and show that it provides an approximation ratio of $(1 + \log K)$, while the approximation ratio for the existing coded delivery schemes is linear in $K$. Numerical simulations show that our proposed algorithm provides a considerable bandwidth reduction over \blue{the} existing coded delivery schemes for almost all content placement schemes.
\end{abstract}


\section{Introduction}

A common way to reduce the burden of traffic in a network system is to take advantage of memories distributed across the network to duplicate parts of the content. This duplication of content is called content placement or caching. A caching system operates in two phases: the placement phase and the delivery phase. In the placement phase which is performed during off-peak hours when network resources are abundant, users fill the local caches with contents anticipating their future demands.
Afterwards, the network is used for an arbitrarily long time, referred to as the delivery phase. This phase can contain a number of rounds where \black{in each round}, users reveal their requests for content and the server must coordinate transmissions such that these requests are satisfied.

Recently, a new class of caching schemes in which the placement and delivery phases are jointly designed (a.k.a "coded caching") has drawn remarkable attention \cite{maddah2014fundamental,maddah2015decentralized}. It has been shown that coded caching can significantly reduce the communication bandwidth requirement for satisfying \black{the} users' requests. Since then, design and analysis of caching techniques for various kinds of networks have been researched extensively 
\cite{niesen2017coded, Zhang_diff_sizes_2015, ji2014order,Qian_exact_rate, Qian_trade_off, Amiri_2016, Wang_2015,karamchandani2016hierarchical, poularakis2016,pedarsani2016online, niesen2015coded, ji2015efficient}.
Most existing works assume that the users follow the prescriptions made by the system for \black{the} content placement. However, it is difficult to enforce all users to follow the particular prescriptions because each user may have its own caching policy or may decide to discard some part of their caches due to lack of space. When the users arbitrarily fill the content of their caches, the existing coded delivery schemes \cite{maddah2014fundamental,maddah2015decentralized, niesen2017coded, ji2015efficient, Zhang_diff_sizes_2015, hueristics_index_coding,ji2014order} could result in very inefficient performance\footnote{\black{Note that although \cite{Zhang_diff_sizes_2015} studies the coded caching problem for the case that the file sizes can be different, the proposed delivery algorithm of this paper is exactly the same as the delivery algorithm of \cite{maddah2014fundamental,maddah2015decentralized, niesen2017coded}.}} (see Example \ref{first_example} in Section \ref{sys_model_and_prob_formulation} and also Theorem \ref{thm:perfromance_SID} for the performance of these schemes under arbitrary caching of users). 

In this paper, we focus on the delivery phase and study a caching system where the content placement phase has been already carried out by the users arbitrarily and has been reported to the server. More specifically, we consider a network consisting of a file server connected through a shared link to $K$ users, each equipped with a cache which has been already filled. In the delivery phase, the users request a set of files. In order to take advantage of multicasting opportunities among the users, each requested file can be divided into a set of subfiles where each subfile is available only in the cache of a group of users \cite{maddah2014fundamental,maddah2015decentralized, niesen2017coded}. The goal is to design a delivery strategy for the server that minimizes the load of the shared link \black{without imposing any restriction on the number of users, size of the files, cache size of the users, and how the placement phase has been carried out.}

This problem is equivalent to a conventional "index coding" problem. In our problem, the total number of subfiles that the server needs to send can be exponential in the number $K$ of users and further, the size of subfiles can be distinct in general. Therefore, existing algorithms \cite{Birk:2006, index_coding_2011,index_coding_LP} for the conventional index coding problem and the recent algorithms for the index coding problems with interlinked cycles \cite{thapa2017interlinked} and symmetric neighboring interference \cite{vaddi2017capacity} are not appropriate for our setting.

Since an index coding problem is computationally hard to solve even only approximately \cite{hardness_network_coding}, in this paper we focus on a specific class of coded multicasting delivery schemes for the server known as the "clique cover delivery scheme". \black{Interestingly, many of the previous works on caching have relied on this class of delivery schemes (for example, see \cite{ji2014order, Shanmugam2016, naderi_ICC_2017, naderi_IT_2017}) due to its practical appeal and its capability in reducing the communication bandwidth compared to uncoded delivery schemes.} In \black{the} clique cover schemes, when a set of subfiles of different files are XORed as a packet and transmitted to the users, for every subfile available in this packet, at least one user requesting it can recover this subfile by using its cache contents and only this XOR transmission. XORing subfiles with different sizes means that
all shorter subfiles are zero padded to match the longest subfile and then XORed. 
If \black{sub-packetization is} allowed, instead of zero padding the shorter subfiles, these subfiles can be padded with bits from other subfiles. Although this kind of schemes could further reduce \black{the} bandwidth usage compared to \black{the} clique cover delivery schemes, they require the knowledge of how the \black{sub-packetization should be performed}. Finding the optimal way of \black{sub-packetization} is formidable in general and only heuristic methods are available \cite{Ramakrishnan_2015, Wan_2017}.

\subsection{Our Contribution}
In this paper, we first formulate the optimal \black{clique cover delivery problem} as a combinatorial optimization problem and show that it can be represented as an Integer Linear Program (ILP). However, the number of variables of this ILP is \black{equal to the number of all cliques of the "side-information graph" \cite{Birk:2006, index_coding_2011,index_coding_LP, Shanmugam2016,ji2014order} (see Remark \ref{rem:side_information} for the detailed description of this graph) generated from the set of all subfiles that the server needs to send.} Since the number of all cliques is generally double exponential in the number $K$ of users, directly solving the ILP is computationally intractable even for a small number of users. 

To overcome the double exponential complexity, we focus on approximation algorithms. We show that the optimal \black{clique cover delivery problem} is equivalent to the weighted set cover problem. Using this connection, we first propose an approximation algorithm which is based on an approximation algorithm for the weighted set cover problem. This algorithm has a good approximation ratio, but its complexity is still double exponential in the number $K$ of users \black{as it requires finding all possible cliques of the side-information graph and also searching over all of these cliques to find a set of cliques that the server should send}. 

Then, we identify features of the optimal \black{clique cover delivery problem} to reduce redundancy in this algorithm and propose Size-Aware Coded Multicast (SACM) algorithm. \black{The significance of SACM algorithm is that it sidesteps the difficulty of finding all possible cliques of the side-information graph and further, it finds the set of cliques that the server should send without the need for searching over all possible cliques.} 
The SACM algorithm has a complexity similar to \black{the} coded delivery schemes of \cite{maddah2014fundamental,maddah2015decentralized,niesen2017coded,ji2014order,hueristics_index_coding,Zhang_diff_sizes_2015} which is generally exponential in the number $K$ of users. We further show that the exponential complexity in \black{the} number $K$ of users is inevitable for any algorithm with good approximation ratio for the optimal \black{clique cover delivery problem}. In terms of performance, we show that SACM algorithm provides $(1 + \log K)$-approximation\footnote{\black{An algorithm is an $\alpha$-approximation to a problem if for any instance of this problem the solution returned by this algorithm is within a factor $\alpha$ of the optimal solution\cite{Combinatorial_Optimization_Korte}.}} 
 which is significantly better than the linear (in $K$) approximation ratio of the existing coded delivery schemes of \cite{maddah2014fundamental,maddah2015decentralized,niesen2017coded,ji2014order,hueristics_index_coding, ji2015efficient,Zhang_diff_sizes_2015}. Furthermore, numerical simulations show that our proposed algorithm provides a considerable bandwidth reduction over the existing coded delivery schemes of \cite{maddah2014fundamental,maddah2015decentralized,niesen2017coded,ji2014order,hueristics_index_coding, ji2015efficient,Zhang_diff_sizes_2015} for almost all content placement schemes.

\subsection{Notation}
For each file $W$, we denote the number of bits or size of $W$ by $\bc{W}$. For two files $W_1$ and $W_2$, bit-wise XOR of $W_1$ and $W_2$ is denoted by $W_1 \oplus W_2$ where the files $W_1$ and $W_2$ are assumed
to be zero padded to match the longest file. 
{\color{black} Sets are denoted by the calligraphic font and sets of sets are denoted by script font. The cardinalities of set $\mathcal{A}$ and set of sets $\mathscr{A}$ are denoted by $|\mathcal{A}|$ and $|\mathscr{A}|$, respectively.}
We use $\vecc(\mathcal{A})$ to denote a column vector with the elements of set $\mathcal{A}$. For a number $N$, we use $[N]$ to denote the set $\{1,2,\ldots,N\}$.

\subsection{Organization}
The rest of this paper is organized as follows. Section \ref{sys_model_and_prob_formulation} describes the system model and 
formally formulates the optimal \black{clique cover delivery problem}. In Section \ref{sec:approximation}, we propose an approximation algorithm for solving the optimal \black{clique cover delivery problem}. Section \ref{clique_cover_description} analyzes the complexity of our proposed algorithm. In Section \ref{sec:num_exp}, we numerically compare our proposed algorithm with the existing coded delivery algorithms. Section \ref{sec:conc} concludes the paper. The proofs of all the technical results of the paper appear in the Appendices.
\section{System Model and Problem Formulation}
\label{sys_model_and_prob_formulation}

\subsection{System Model}
\label{model}
We consider a system with one server connected through a shared, error-free link to $K$ users as shown in Fig. \ref{fig:model}. 
The server has access to a database of $N$ popular files 
$W_1,W_2,\ldots,W_N$ where the size of file $W_n$ in bits is denoted by $\bc{W_n}$. We use the notations $[K]:= \{1,2,\ldots,K\}$ and $[N]:= \{1,2,\ldots,N\}$.
Each user $k$ has a separate cache memory $Z_k$. We assume that users have filled the content of their caches using the database in the \textit{Placement phase}. 

In the \textit{Delivery phase}, each user $k$ requests a file index $d_k \in [N]$. 
We assume that users request different files\footnote{Note that the probability of each user requesting a distinct file goes to one as $N \longrightarrow \infty$ for a fixed number $K$ of users. Hence, this assumption holds with high probability because it is likely that $N \gg K$ in practice. Further, note that the proposed algorithm of this paper can still be applied to the case where there is repetition in the users' requests by pretending that they are different.}. Without loss of generality, suppose user $k$ requests file $k$ in the delivery phase, that is, $d_k =k$.

\begin{figure}
\begin{center}
\resizebox{0.27\textwidth}{!}{%
\begin{tikzpicture}
\node[rectangle,rounded corners,draw, inner sep=0pt,fill=black!32, thick, minimum width=2cm,
                        minimum height = 1cm] (S) at (0,0) {Server};

\node[rectangle,rounded corners,draw, inner sep=0pt,fill=black!4, thick, minimum width=2cm,
                        minimum height = 1cm] (U_1) at (-4,-4) {User $1$};  
                        
\node[rectangle,rounded corners,draw, inner sep=0pt,fill=black!4, thick, minimum width=2cm,
                        minimum height = 1cm] (U_2) at (-1,-4) {User $2$};                                                
 
\node[rectangle,rounded corners,draw, inner sep=0pt,fill=black!4, thick, minimum width=2cm,
                        minimum height = 1cm] (U_N) at (3,-4) {User $K$};   
                        
\node[rectangle,rounded corners,draw, inner sep=0pt,fill=black!14, thick, minimum width=1cm,
                        minimum height = 1cm] (C_1) at (-4,-5.2) {$Z_1$};  
                        
\node[rectangle,rounded corners,draw, inner sep=0pt,fill=black!14, thick, minimum width=1cm,
                        minimum height = 1cm] (C_2) at (-1,-5.2) {$Z_2$};                                                
 
\node[rectangle,rounded corners,draw, inner sep=0pt,fill=black!14, thick, minimum width=1cm,
                        minimum height = 1cm] (C_N) at (3,-5.2) {$Z_K$};  
                        
\node[rectangle,rounded corners,draw, inner sep=0pt,fill=black!14, thick, minimum width=0.6cm,
                        minimum height = 1cm] (F_1) at (-1.2,1.2) {$W_1$};  
\node[rectangle,rounded corners,draw, inner sep=0pt,fill=black!14, thick, minimum width=0.6cm,
                        minimum height = 1cm] (F_2) at (-.6,1.2) {$W_2$};       
\node[rectangle,rounded corners,draw, inner sep=0pt,fill=black!14, thick, minimum width=0.6cm,
                        minimum height = 1cm] (F_3) at (0.0,1.2) {$W_3$};       
\node[rectangle,rounded corners,draw, inner sep=0pt,fill=black!14, thick, minimum width=0.6cm,
                        minimum height = 1cm] (F_3) at (1.4,1.2) {$W_N$};

 \draw[line width = 1.4mm, dash pattern=on .05mm off 2mm,
                                         line cap=round] (0.8,-4) -- (1.3,-4); 
                                         
  \draw[line width = 1.0mm, dash pattern=on .05mm off 2mm,
                                         line cap=round] (0.5,1.2) -- (1.0,1.2);                                                               
\draw [line width=0.4mm, black ] (S.south)--(0,-2);
\draw [line width=0.4mm, black ] (0,-2)--(U_1.north);
\draw [line width=0.4mm, black ] (0,-2)--(U_2.north);      
\draw [line width=0.4mm, black ] (0,-2)--(U_N.north);     
\end{tikzpicture}
}
\end{center}
\caption{Illustration of system model.}
\label{fig:model}
\end{figure}
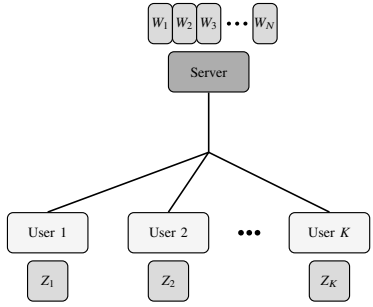
Upon receiving the requests  $\boldsymbol{d} = (d_1,\ldots,d_K) = (1,\ldots,K)$ of the users, the server responds by transmitting a message $\boldsymbol U = (U_1,U_2,\dots)$ that consists of a sequence of packets over the shared link. Each user $k$ then aims to reconstruct its requested file $W_{k}$ using this message and its cache contents.
\subsection{\black{Coded Multicasting Delivery Strategy}}
\label{info_stra}
We focus on the delivery phase in this paper. We assume that the server knows the requests $\boldsymbol{d} = (1,\ldots,K)$, the database of files $W_1,W_2,\ldots,W_N$, and the contents $Z_1,Z_2,\ldots,Z_K$ of the cache of users from the placement phase\footnote{In practice, each file $W_n$ is divided \black{into} a number of parts with unique sub-indices and a user can cache some (zero or more) of these parts. By communicating
only the sub-indices for all files that a user has cached to the server, the server can reconstruct the cache
contents of this user.}.
In order to take advantage of multicasting opportunities among the users and as proposed in the coded caching litertaure (e.g., see \cite{maddah2014fundamental,maddah2015decentralized, niesen2017coded}), 
given the cache content of all the users, the bits of the files can be grouped into subfiles $W_{n,\mathcal{A}}$, where $W_{n,\mathcal{A}}$ is the set of bits of file $W_n$ which is available only in the cache of users belonging to the set $\mathcal{A}$. For example, if $K=3$, then for each $n \in [N]$, file $W_n$ can be split into $W_{n,\varnothing}, W_{n,\{1\}}, W_{n,\{2\}},W_{n,\{3\}}, W_{n,\{1,2\}}, W_{n,\{1,3\}}, W_{n,\{2,3\}},$ $W_{n,\{1,2,3\}}$. 

For each user $k \in [K]$, the server needs to send all non-empty subfiles of file $W_{k}$ which is not available in the cache of user $k$. Let
\begin{align}
\mathcal{W}_k = \{ W_{k, \mathcal{A}}: \mathcal{A} \subseteq [K]\setminus \{k\}, \bc{W_{k, \mathcal{A}}} \neq 0 \}.
\end{align}
Then, $\mathcal{W}_k$ is the set of all subfiles that needs to be sent by the server to user $k$. For example, when $K=3$, we have $\mathcal{W}_1 = \{W_{1,\varnothing}, W_{1,\{2\}}, W_{1,\{3\}}, W_{1,\{2,3\}} \}$ if all these subfiles have non-zero sizes.
We denote the set of all subfiles that needs to be sent by the server as $\mathcal{W} = \cup_{k \in [K]} \mathcal{W}_k$ and we define $\tau_K = |\mathcal{W}|$.

{\color{black}
In this paper, we focus on a specific class of coded multicasting delivery schemes known as the clique\footnote{A clique is a subset of vertices of an undirected graph such that every two distinct vertices in the clique are adjacent. In other words, a clique is a graph where every vertex is adjacent to every other.} cover scheme \cite{ji2014order}. Clique cover delivery schemes have the following property.
\begin{property}
\label{delivery_property}
When a subset $\mathcal{P}$ of $\mathcal{W}$ \black{is} XORed, for every subfile in $\mathcal{P}$, at least one user requesting it can recover this subfile by using its cache contents and only XOR transmission $U=  \oplus_{W \in \mathcal{P}} W$.
\end{property}
For example, if $A \oplus B \oplus C$ was sent, there is a user wanting $A$ that could recover it by using $B$ and $C$ stored in its cache and similar conditions hold for one user wanting $B$ and one user wanting $C$.
We call any $\mathcal{P} \subseteq \mathcal{W}$ that satisfies Property \ref{delivery_property} a \textit{feasible packet}. We interchangeably use the term feasible packet to refer to both $\mathcal{P}$ and $U=  \oplus_{W \in \mathcal{P}} W$ satisfying Property \ref{delivery_property}.
The size or the number of bits of a packet $U =  \oplus_{W \in \mathcal{P}} W$ (or packet $\mathcal{P}$) is denoted by $\bc{U}$ (respectively by $\bc{\mathcal{P}}$) and is given by 
 \begin{align}
 \label{bc_definition}
 \bc{U} =\bc{\mathcal{P}} = \max_{W \in \mathcal{P}} \bc{W}.
 \end{align}
\begin{rem}
\label{rem:side_information}
The reason that a delivery scheme with Property 1 is called a clique cover scheme is as follows: Consider a graph $\mathcal{G}_c$ where the set of vertices is $\mathcal{W}$ and the weight of vertex $W_{k,\mathcal{A}} \in \mathcal{W}$ is $\bc{W_{k,\mathcal{A}}}$. In this graph, there is an edge between $W_{k,\mathcal{A}}$ and $W_{l,\mathcal{B}}$ if and only if subfiles $W_{k,\mathcal{A}}$ and $W_{l,\mathcal{B}}$ are stored in the cache of users $l$ and $k$, respectively. Then, for this graph, known as "side-information graph" in the literature \cite{Birk:2006, index_coding_2011,index_coding_LP, Shanmugam2016,ji2014order}, each clique is a feasible packet and vice-versa. Hereafter, we use the terms feasible packet and clique interchangeably. 
\end{rem}

Note that Property 1 has two straightforward implications: 1) a clique (or feasible packet) $\mathcal{P}$ cannot include more than one subfile requested by each user $k \in [K]$ and hence, it can be written as a set $\{W_{k,\mathcal{A}_k}, k \in \mathcal{M}\}$ for some $\mathcal{M} \subseteq [K]$; 2) for each user $k \in \mathcal{M}$ to be able to recover $W_{k,\mathcal{A}_k}$, it should have $W_{k',\mathcal{A}_{k'}}$ for all $k' \in \mathcal{M} \setminus \{k\}$ stored in its cache, that is, for all $k' \in \mathcal{M} \setminus \{k\}$, $\mathcal{A}_{k'}$ should include $k$ (or equivalently, $k \in \cap_{k' \in  \mathcal{M} \setminus \{k\}} \mathcal{A}_{k'}$). Such a clique $\mathcal{P}$ can be equivalently described as  $\mathcal{P}_{\mathcal{A}_{1:K}}$ with the convention that $\mathcal{A}_k=\aleph$ if $k \notin \mathcal M$.
For example, in the case of  $K=3$, $\mathcal{P}_{\{2,3\},\{1,3\},\aleph} = \{ W_{1,\{2,3\}}, W_{2,\{1,3\}}\}$ and
$\mathcal{P}_{\aleph,\{3\},\aleph} = \{ W_{2,\{3\}}\}$. 
}
\subsection{Problem Formulation}
{\color{black}
The aim of this paper is to find a clique cover delivery scheme that minimizes the total number of bits that the server needs to send to the users such that each user is able to reconstruct the file it has requested. This requires each subfile belonging to the set $\mathcal{W}$ to be sent at least once. Let $\mathscr{P}$ denote the set of all feasible packets that can be generated from the set $\mathcal{W}$ of all subfiles (equivalently, $\mathscr{P}$ is the set of all cliques of the side-information graph $\mathcal{G}_c$) and let $\lambda_K = |\mathscr{P}|$. For any $\mathcal{P} \in \mathscr{P}$, let $\alpha_{\mathcal{P}} \in \{0,1\}$ be a variable indicating whether $\mathcal{P}$ is selected for transmission. We call $\boldsymbol{\alpha} = \vecc(\{\alpha_{\mathcal{P}}: \mathcal{P} \in \mathscr{P}\})$ a clique cover delivery scheme.
The problem of designing the optimal clique cover scheme is formally defined below.

\begin{problem}[Optimal Clique Cover Delivery Problem]
\label{optimal_delivery_N}
For a system of $K$ users, given the set $\mathcal{W}$ of all subfiles that needs to be sent by the server, find a clique cover delivery scheme $\boldsymbol{\alpha}$ that solves
\begin{align*}
\min & \sum_{\mathcal{P} \in \mathscr{P}} \alpha_{\mathcal{P}} \bc{\mathcal{P}}\\
s.t. \quad &\cup_{\mathcal{P} \in \mathscr{P}: \alpha_{\mathcal{P}}=1} \mathcal{P} = \mathcal{W},
\end{align*}
where 
$\mathscr{P}$ is the set of all cliques that can be generated from $\mathcal{W}$.
\end{problem}

\begin{rem}
Note that although in Problem \ref{optimal_delivery_N}, we want to find a clique cover with \black{the} minimum sum of sizes of cliques, this problem is different from "Minimum Weighted Clique Covering Problem" defined in the literature \cite{grotschel2012geometric}.
\end{rem}
}
The optimal \black{clique cover delivery problem} can be represented as an Integer Linear Program (ILP). To this end, recall that $\tau_K = |\mathcal{W}|$ and $\lambda_K = |\mathscr{P}|$. We define $L$ to be a $\{0,1\}$-valued matrix with $\tau_K$ rows and $\lambda_K$ columns where each row corresponds to a subfile that needs to be sent to a user and each column corresponds to one \black{clique}. 
The entry of $L$ corresponding to subfile $W$ and \black{clique} $\mathcal{P}$ is $1$ if $W$ can be decoded from $\mathcal{P}$; otherwise, it is $0$. Using the matrix $L$, the condition that each subfile belonging to the set $\mathcal{W}$ should be sent at least once can be written as $L \boldsymbol{\alpha} \geq \textbf{1}$ where $\textbf{1}$ is an all-one vector of size $\tau_K$. Then, the ILP can be described as follows.

\begin{problem}
\label{ILP_problem}
For a system of $K$ users, given the set $\mathcal{W}$ of all subfiles and the set of $\mathscr{P}$ all \black{cliques} that can be generated from $\mathcal{W}$, find a \black{clique cover delivery scheme} $\boldsymbol{\alpha}$ that solves
\begin{align*}
\min & \sum_{\mathcal{P} \in \mathscr{P}} \alpha_{\mathcal{P}} \bc{\mathcal{P}}\\
s.t. \quad &L \boldsymbol{\alpha} 
 \geq \textbf{1}.
\end{align*}

\end{problem}
We end this section by providing a motivating example that shows the importance of solving \black{the} optimal \black{clique cover delivery problem} (Problem \ref{optimal_delivery_N}). This simple example indicates the significant bandwidth reduction that can be \black{achieved} by solving Problem \ref{optimal_delivery_N} compared to the conventional uncoded delivery, Greedy Coded Multicast (GCM) scheme \cite{maddah2014fundamental,maddah2015decentralized,niesen2017coded,Zhang_diff_sizes_2015}, \black{GCC scheme \cite{ji2014order}, GCLC and HgLC schemes \cite{ji2015efficient}}, and Graph Coloring-based Coded Multicast (GCCM) scheme \cite{hueristics_index_coding}. Note that in the conventional uncoded delivery, \black{the} server sends each subfile of set $\mathcal{W}$ separately. 
\begin{examp}
\label{first_example}
Let $K=3$ and $\mathcal{W} = \{W_{1,\varnothing}, W_{1,\{2\}}, W_{1,\{3\}},$ 
$W_{1,\{2,3\}}, W_{2,\varnothing}, W_{2,\{1\}}, W_{2,\{3\}}, W_{2,\{1,3\}}, W_{3,\varnothing}, W_{3,\{1\}}, W_{3,\{2\}},$ $W_{3,\{1,2\}}\}$ where the size of subfiles $W_{1,\{2,3\}}$ and $W_{2,\{1\}}$ is $300$ bits and the rest of subfiles in $\mathcal{W}$ have the size of $10$ bits.
In this case, \black{the uncoded delivery scheme sends 10 subfiles of size 10 bits and 2 subfiles of size 300 bits, resulting in the total number of bits of $10 \times 10 + 2 \times 300 = 700$}. The GCM, \black{GCC, GCLC, and HgLC schemes
choose}\footnote{\black{Note that while these schemes are generally different, under our assumption that users request different files, they all will simplify to the same algorithm.}} \black{cliques} $\mathcal{P}_{\{2,3\}, \{1,3\}, \{1,2\}}$, $\mathcal{P}_{\{2\},\{1\}, \aleph}$, $\mathcal{P}_{\{3\},\aleph, \{1\}}$, $\mathcal{P}_{\aleph, \{3\},\{2\}}$, $\mathcal{P}_{\varnothing,\aleph, \aleph}$,$\mathcal{P}_{\aleph,\varnothing, \aleph}$,
$\mathcal{P}_{\aleph, \aleph, \varnothing}$ \black{with sizes equal to 300, 300, 10, 10, 10, 10, 10; respectively. Therefore, the total number of bits sent by the GCM,\black{GCC, GCLC, and HgLC schemes} is $2 \times 300+ 5 \times 10 = 650$}. The GCCM scheme\black{\footnote{Note that this algorithm picks vertices arbitrarily and hence, there may be many outputs for this algorithm. We pick one arbitrarily. See \cite{hueristics_index_coding} for more details.}} chooses \black{cliques} $\mathcal{P}_{\aleph, \{3\},\{2\}}$, $\mathcal{P}_{\{2\}, \{1,3\}, \aleph}$, $\mathcal{P}_{\{2,3\}, \aleph,\{1,2\}}$, $\mathcal{P}_{\{3\},\aleph, \{1\}}$, $\mathcal{P}_{\aleph, \{1\},\aleph}$, $\mathcal{P}_{\varnothing,\aleph, \aleph}$,$\mathcal{P}_{\aleph,\varnothing, \aleph}$,
$\mathcal{P}_{\aleph, \aleph, \varnothing}$ \black{with sizes equal to 10, 10, 300, 10, 300, 10, 10, 10; respectively. Therefore, the total number of bits sent by the GCCM scheme is $2 \times 300+ 6 \times 10 = 660$.}
If we solve\footnote{The optimal solution to the ILP is calculated using the GUROBI optimization solver \cite{gurobi}.} Problem
\ref{ILP_problem}, the \black{cliques}
$\mathcal{P}_{\{2\},\{1,3\}, \aleph}$, $\mathcal{P}_{\{3\},\aleph, \{1\}}$, $\mathcal{P}_{\aleph, \{3\},\{2\}}$, $\mathcal{P}_{\aleph, \aleph,\{1,2\}}$, $\mathcal{P}_{\varnothing,\aleph, \aleph}$, $\mathcal{P}_{\aleph,\varnothing, \aleph}$,
$\mathcal{P}_{\aleph, \aleph, \varnothing}$, $\mathcal{P}_{\{2,3\},\{1\}, \aleph}$ are chosen \black{with sizes equal to 300, 10, 10, 10, 10, 10, 10, 10; respectively. Therefore, the total number of bits sent by the optimal solution is $300+ 7 \times 10 = 370$}.
\black{This indicates that \black{the} \black{optimal clique cover delivery strategy} reduces the number of bits required to be sent by 47\% compared to \black{the} uncoded delivery while the GCM (and similarly \black{GCC, GCLC, and HgLC}) and GCCM schemes can decrease this number only by 7.1\% and 5.7\%, respectively.}
\end{examp}
As can be seen from this simple example, the uncoded delivery is inefficient because it does not take advantage of sending multiple subfiles together as a packet. Furthermore, the GCM (and similarly \black{GCC, GCLC, and HgLC}) and GCCM schemes are not able to choose these packets efficiently enough because they do not take the size of subfiles into consideration. 

\section{An Approximation Algorithm for \black{Optimal Clique Cover Delivery Problem}}
\label{sec:approximation}
In Section \ref{sys_model_and_prob_formulation}, we formulated the \black{optimal clique cover delivery problem} as a combinatorial optimization problem and we further showed that it can be represented as the ILP of Problem \ref{ILP_problem}.
However, this ILP suffers from two significant challenges.
\begin{enumerate}[label=(\mylabel{P}{\arabic*})]
\item In order to solve the ILP, one should first find the set $\mathscr{P}$ of all feasible packets that can be generated from the set $\mathcal{W}$ of subfiles. However, finding $\mathscr{P}$ is the same as finding all possible cliques of \black{the side-information graph} (\black{a problem that is} known as "Clique Problem") which is NP-hard \cite{karp1972reducibility}.

\item The number $\tau_K$ of subfiles \black{that needs} to be sent by the server can be as large as $O(2^K)$ and the number $\lambda_K$ of all \black{cliques} can be as large as $O(2^{\tau_K})$. Hence, the number of variables in the ILP can be $O(2^{2^K})$, that is, double exponential in $K$.
\end{enumerate}
Due to the above points, directly solving this ILP is computationally intractable as its complexity is double exponential in \black{the} number $K$ of users. In this section, we propose an approximation algorithm for Problem \ref{optimal_delivery_N}. This algorithm provides the approximation ratio of $(1 + \log K)$ and it has a complexity similar to the coded delivery schemes of \cite{maddah2014fundamental,maddah2015decentralized,niesen2017coded,ji2014order,hueristics_index_coding,Zhang_diff_sizes_2015}. 
To this end, we show that \black{the} \black{optimal clique cover delivery problem} (Problem \ref{optimal_delivery_N}) is equivalent to a weighted set cover problem \cite{Combinatorial_Optimization_Korte}. This equivalence can be easily observed by considering $\mathcal{W}$ as a set of elements that needs to be covered in the weighted set cover problem, $\mathscr{P}$ as a set of subsets of the elements' set in the weighted set cover problem, and for each $\mathcal{P} \in \mathscr{P}$, $\bc{\mathcal{P}}$ as the weight of subset $\mathcal{P}$ of elements in the weighted set cover problem.
%
%
%
%
Now using this equivalence of problems, we propose Algorithm \ref{first_delivery} for solving Problem \ref{optimal_delivery_N}. This algorithm is based on an approximation algorithm for the weighted set cover problem \cite{chvatal1979greedy} which has been modified according \black{to} the following property of the set $\mathscr{P}$ of all feasible packets in Problem \ref{optimal_delivery_N}: \black{If $\mathcal{P}_1$ and $\mathcal{P}_2$ are feasible packets, then $\mathcal{P}_1 \setminus \mathcal{P}_2$ is either an empty packet or a feasible packet. The correctness of this property \black{results} from the fact that all feasible packets are cliques of the side-information graph and hence, subtracting the vertices of one clique from another one results in either another clique or an empty set.}

\black{Before presenting the results about the approximation ratio that Algorithm \ref{first_delivery} provides for Problem \ref{optimal_delivery_N}, we first describe this algorithm in simple words: Let $\mathcal{E}$ be any subset of \black{the} set $\mathcal{W}$ of all subfiles that the server needs to send. Furthermore, let $\mathscr{S}$ be the set of all cliques that can be generated from $\mathcal{E}$. Then, in each iteration of Algorithm \ref{first_delivery}, the \textsc{Packet-Based-Optimizer (PBO)} function determines a clique that maximizes $\frac{|\mathcal{P}|}{\bc{\mathcal{P}}}$ which is a metric that measures the number of subfiles included in clique $\mathcal{P}$ per bit unit. Note that this metric captures the fact that we would like to send as many subfiles as possible with the minimum number of bits.
Let $\mathcal{P^{\black{*}}}$ denote this clique. Then, we add this clique to set $\mathscr{C}$ of all cliques that the server needs to send. Since it is redundant to send a subfile more than once, we remove the subfiles available in \black{the} clique $\mathcal{P^{\black{*}}}$ from the set $\mathcal{E}$. At the end, we update the set $\mathscr{S}$ by eliminating all cliques $\mathcal{S}$ that intersect with $\mathcal{P^{\black{*}}}$.}

\begin{algorithm}
\caption{}
\label{first_delivery}
\hspace*{\algorithmicindent} {\small \textbf{Input:} Set $\mathcal{W}$ of subfiles and set $\mathscr{P}$ of all cliques.}
\begin{algorithmic}[1]
\State $\mathscr{C} = \emptyset$
\State $\mathcal{E} = \mathcal{W}$
\State $\mathscr{S} = \mathscr{P}$
\While{$\mathcal{E} \neq \emptyset$}
\State $\mathcal{P^{\black{*}}} =$ \textsc{PBO}$(\mathscr{S})$
\State $\mathscr{C} = \mathscr{C} \cup \{\mathcal{P^{\black{*}}}\}$
\State $\mathcal{E} = \mathcal{E} \setminus \mathcal{P^{\black{*}}}$
\For{$\mathcal{S} \in \mathscr{S}$}
\If{$\mathcal{S} \cap \mathcal{P^{\black{*}}} \neq \emptyset$}
\State $\mathscr{S} = \mathscr{S} \setminus \{ \mathcal{S}\}$
\EndIf
\EndFor
\EndWhile
\end{algorithmic}
\hspace*{\algorithmicindent} {\small \textbf{Output:} Set of cliques $\mathscr{C}$.}
\end{algorithm} 
\begin{algorithm}
\begin{algorithmic}
\Function{PBO}{$\mathscr{S}$}
  \State \Return $\argmax_{\mathcal{P} \in \mathscr{S}} \dfrac{|\mathcal{P}|}{\bc{\mathcal{P}}}$
\EndFunction
\end{algorithmic}
\end{algorithm} 
\begin{lem}
\label{optimal_delivery_solution_first}
If the set $\mathscr{P}$ of all \black{cliques} is known, Algorithm \ref{first_delivery} achieves a $(1 + \log K)$-approximation to Problem~\ref{optimal_delivery_N}.
\end{lem}
\begin{proof}
See Appendix \ref{proof:optimal_delivery_solution_first}.
\end{proof}

Note that although Lemma \ref{optimal_delivery_solution_first} suggests an algorithm with a good approximation ratio for Problem \ref{optimal_delivery_N}, it still suffers from issue (P1). \black{Furthermore, in each iteration of Algorithm \ref{first_delivery}, in order to find the output of function PBO, one should go over the set 
$\mathscr{S}$ of cliques and find the clique maximizing $\frac{|\mathcal{P}|}{\bc{\mathcal{P}}}$. The number of cliques in set $\mathscr{S}$ can be generally $O(2^{2^K})$, that is, double exponential in the number $K$ of users.
To sidestep these difficulties, in the next section, we propose a new algorithm that enjoys the good approximation ratio of Algorithm \ref{first_delivery}, but it does not work with the set $\mathscr{S}$ of cliques.}
\subsection{Size-Aware Coded Multicast}
In this section, we first propose \textsc{Subfile-Based-Optimizer (SBO)}, an alternative way of calculating function PBO which only needs the set $\mathcal{W}$ of subfiles. The following theorem states this result.
\begin{thm}
\label{equivalent_function}
For any subset $\mathcal{E} \subseteq \mathcal{W}$ of subfiles, let $\mathscr{S}$ denote the set of all \black{cliques} that can be generated from $\mathcal{E}$. Then, the size of the \black{cliques} obtained by functions PBO and SBO \black{is} equal, that is, 
$\bc{\textsc{PBO}(\mathscr{S})} = \bc{\textsc{SBO}(\mathcal{E})}$.
\end{thm}
\begin{proof}
See Appendix \ref{proof:optimal_delivery_solution_second}.
\end{proof}
\begin{algorithm}[h]
\begin{algorithmic}
\Function{SBO}{$\mathcal{E}$} 
\State Let $\mathcal{M} = \{k: k \in [K], \exists W_{k, \mathcal{A}} \in \mathcal{E} \text{ for some } \mathcal{A} \subseteq [K] \setminus \{k\} \}$
\black{\State Let $\mathscr{T} = \{\mathcal{T}: \mathcal{T} \subseteq \mathcal{M}, |\mathcal{T}|>0 \}$}
\For{\black{$\mathcal{T} \in \mathscr{T}$}}
	\For{$j \in \mathcal{T}$}
		\State Let $\mathcal{L}_{j,\mathcal{T}} = \{W_{j, \mathcal{A}}: W_{j, \mathcal{A}} \in \mathcal{E}, \mathcal{T} \setminus \{j\}  \subseteq \mathcal{A} \}$\footnotemark
		\State Calculate $V_{j,\mathcal{T}} = \argmin_{W \in \mathcal{L}_{j,\mathcal{T}}} \bc{W}$
	\EndFor
	\State Let $\mathcal{R}_{\mathcal{T}} = \{V_{k,\mathcal{T}} : k \in \mathcal{T} \}$
\EndFor
\State  \Return $\argmax_{\mathcal{R}_{\mathcal{T}} : \black{\mathcal{T} \in \mathscr{T}}} \dfrac{|\mathcal{T}|}{\bc{\mathcal{R}_{\mathcal{T}} }}$\footnotemark
\EndFunction
\end{algorithmic}
\end{algorithm} 
\addtocounter{footnote}{-2} 
 \stepcounter{footnote}\footnotetext{In case that there are more than one subfile $W \in \mathcal{L}_{j,\mathcal{T}}$ minimizing $\bc{W}$, we can pick one arbitrarily. Further, define $\mathcal{L}_{j,\mathcal{T}} = \{W_*\}$ with $\bc{W_*} = \infty$ whenever $\mathcal{L}_{j,\mathcal{T}}$ is empty.}
 \stepcounter{footnote}\footnotetext{In case that there are more than one packet $\mathcal{R}_{\mathcal{T}}$ maximizing $\frac{|\mathcal{T}|}{\bc{\mathcal{R}_{\mathcal{T}}}}$, we can pick one arbitrarily.}
{\color{black}

\begin{examp}
\label{example:SBO_PBO}
Consider Example \ref{first_example} with $K=3$ and $\mathcal{W} = \{W_{1,\varnothing}, W_{1,\{2\}}, W_{1,\{3\}}, W_{1,\{2,3\}}, W_{2,\varnothing}, W_{2,\{1\}}, W_{2,\{3\}}, W_{2,\{1,3\}}, W_{3,\varnothing},$ $W_{3,\{1\}},$ $W_{3,\{2\}}, W_{3,\{1,2\}}\}$ where the size of subfiles $W_{1,\{2,3\}}$ and $W_{2,\{1\}}$ is $300$ bits and the rest of subfiles in $\mathcal{W}$ have the size of $10$ bits. The side-information graph $\mathcal{G}_c$ of this example is shown in Fig. \ref{grpah_ex1}. Now, let $\mathcal{E} = \mathcal{W}$ and $\mathscr{S}$ be the set of all feasible packets that can be generated from $\mathcal{W}$ which is the set of all cliques of graph $\mathcal{G}_c$ in Fig. \ref{grpah_ex1}.
\begin{small}
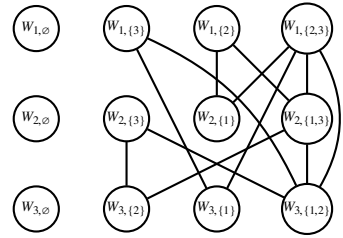
\begin{figure}
\begin{center}
\begin{tikzpicture}[minimum size=10mm, node distance=1cm, every loop/.style={},
                    thick,state/.style={circle, inner sep=0pt, draw,font=\sffamily\small\bfseries}]
\begin{scope} [rotate=-90, scale=0.6, every node/.append style={scale=0.6}]
\node[state] (A_1) at (-4,0) {$W_{1,\{2,3\}}$};
\node[state] (B_1) at (-4,-2) {$W_{1,\{2\}}$};
\node[state] (C_1) at (-4,-4) {$W_{1,\{3\}}$};
\node[state] (D_1) at (-4,-6) {$W_{1,\varnothing}$};

\node[state] (A_2) at (-2,0) {$W_{2,\{1,3\}}$};
\node[state] (B_2) at (-2,-2) {$W_{2,\{1\}}$};
\node[state] (C_2) at (-2,-4) {$W_{2,\{3\}}$};
\node[state] (D_2) at (-2,-6) {$W_{2,\varnothing}$};

\node[state] (A_3) at (0,0) {$W_{3,\{1,2\}}$};
\node[state] (B_3) at (0,-2) {$W_{3,\{1\}}$};
\node[state] (C_3) at (0,-4) {$W_{3,\{2\}}$};
\node[state] (D_3) at (0,-6) {$W_{3,\varnothing}$};

\draw[every node/.style={font=\sffamily\small}]
(A_1)  edge node {}   (A_2)
(A_1)  edge node {}   (B_2)
(A_1)  edge node {}   (B_3)

(A_2)  edge node {}   (A_3)
(A_2)  edge node {}   (B_1)
(A_2)  edge node {}   (C_3)

 (A_1) edge [bend left] node {}   (A_3)
(A_3)  edge [bend right =20] node {}   (C_1)
(A_3)  edge  node {}   (C_2)

(B_1)  edge node {}   (B_2)
(C_1)  edge node {}   (B_3)
(C_2)  edge node {}   (C_3);
\end{scope}
\end{tikzpicture}
\end{center}
\caption{\black{Side-information graph $\mathcal{G}_c$ of Example \ref{example:SBO_PBO}.}}
\label{grpah_ex1}
\end{figure}
\end{small}
\begin{itemize}
\item How to calculate $\textsc{PBO}(\mathscr{S})$? To this end, one should go over the set $\mathscr{S}$ of all cliques of graph $\mathcal{G}_c$ and calculates $\frac{|\mathcal{P}|}{\bc{\mathcal{P}}}$ for each clique $\mathcal{P} \in \mathscr{S}$ to find the clique maximizing this ratio. The set $\mathscr{S}$ is given as follows,

\begin{small}
\begin{align}
&\hspace{-7pt} \mathscr{S}= \{ \mathcal{P}_{\{2,3\}, \{1,3\}, \{1,2\}},
\mathcal{P}_{\{2,3\},\{1,3\},\aleph},
\mathcal{P}_{\aleph,\{1,3\},\{1,2\}},
\mathcal{P}_{\{2,3\},\aleph,\{1,2\}},
\notag \\
&\hspace{-1pt} \mathcal{P}_{\{2,3\},\{1\},\aleph},
\mathcal{P}_{\{2,3\},\aleph, \{1\}},
\mathcal{P}_{\{2\},\{1,3\},\aleph}, 
\mathcal{P}_{\aleph,\{1,3\},\{2\}}, 
\mathcal{P}_{\{3\}, \aleph, \{1,2\}},
\notag \\
&\hspace{-1pt}
\mathcal{P}_{\aleph,\{1,3\},\{2\}}, 
\mathcal{P}_{\{3\}, \aleph, \{1,2\}},
\mathcal{P}_{\aleph, \{3\}, \{1,2\}},
\mathcal{P}_{\{2\},\{1\}, \aleph},
\mathcal{P}_{\{3\},\aleph, \{1\}},
\notag \\
&\hspace{-1pt}
\mathcal{P}_{\aleph, \{3\},\{2\}},
\mathcal{P}_{\{2,3\}, \aleph, \aleph},
\mathcal{P}_{\aleph, \{1,3\}, \aleph},
\mathcal{P}_{\aleph, \aleph, \{1,2\}},
\mathcal{P}_{\{2\},\aleph, \aleph},
\mathcal{P}_{\{3\},\aleph, \aleph},
\notag \\
&
\mathcal{P}_{\aleph, \{1\}, \aleph},
\mathcal{P}_{\aleph, \{3\}, \aleph}, 
\mathcal{P}_{\aleph, \aleph, \{1\}},
\mathcal{P}_{\aleph, \aleph, \{2\}},
\mathcal{P}_{\varnothing,\aleph, \aleph},
\mathcal{P}_{\aleph,\varnothing, \aleph},
\mathcal{P}_{\aleph, \aleph, \varnothing} \}.
\end{align}
\end{small}

\item Why $\bc{\textsc{PBO}(\mathscr{S})} = \bc{\textsc{SBO}(\mathcal{E})}$? First note that the set $\mathscr{S}$ of all cliques of graph $\mathcal{G}_c$ can be decomposed into following 7 disjoint sets $\mathscr{Q}_{\mathcal{T}}$ for 
$\mathcal{T} \in \mathscr{T}:= \{\{1\},\{2\}, \{3\},\{1,2\},\{1,3\},$ $\{2,3\}, \{1,2,3\}\}$, that is $\mathscr{S} = \cup_{\mathcal{T} \in \mathscr{T}} \mathscr{Q}_{\mathcal{T}}$, where $\mathscr{Q}_{\mathcal{T}}$ is the set of all cliques $\mathcal{P}$ including exactly one subfile for each user in set $\mathcal{T}$ and no subfile for users not in set $\mathcal{T}$ (See Appendix \ref{proof:optimal_delivery_solution_second} for a proof of disjointness of sets $\mathscr{Q}_{\mathcal{T}}$'s).

\begin{small}
\begin{align}
& \mathscr{Q}_{\{1\}} = \{\mathcal{P}_{\varnothing,\aleph, \aleph}, \mathcal{P}_{\{2\},\aleph, \aleph},
\mathcal{P}_{\{3\},\aleph, \aleph}, \mathcal{P}_{\{2,3\}, \aleph, \aleph}\}, \notag \\
& \mathscr{Q}_{\{2\}} = \{ \mathcal{P}_{\aleph,\varnothing, \aleph}, \mathcal{P}_{\aleph, \{1\}, \aleph},
\mathcal{P}_{\aleph, \{3\}, \aleph}, \mathcal{P}_{\aleph, \{1,3\}, \aleph} \}, \notag \\
&\mathscr{Q}_{\{3\}} = \{\mathcal{P}_{\aleph, \aleph, \varnothing}, \mathcal{P}_{\aleph, \aleph, \{1\}},
\mathcal{P}_{\aleph, \aleph, \{2\}}, \mathcal{P}_{\aleph, \aleph, \{1,2\}}\}, \notag \\
& \mathscr{Q}_{\{1,2,3\}} = \{  \mathcal{P}_{\{2,3\}, \{1,3\}, \{1,2\}} \}, \notag \\
&\mathscr{Q}_{\{1,2\}}= \{\mathcal{P}_{\{2,3\},\{1,3\},\aleph}, \mathcal{P}_{\{2,3\},\{1\},\aleph}, \mathcal{P}_{\{2\},\{1,3\},\aleph}, \mathcal{P}_{\{2\},\{1\}, \aleph}\}, \notag \\
&\mathscr{Q}_{\{1,3\}} = \{\mathcal{P}_{\{2,3\}, \aleph,\{1,2\}}, \mathcal{P}_{\{2,3\},\aleph, \{1\}}, \mathcal{P}_{\{3\}, \aleph, \{1,2\}},  \mathcal{P}_{\{3\},\aleph, \{1\}} \}, \notag \\
&\mathscr{Q}_{\{2,3\}} = \{\mathcal{P}_{\aleph,\{1,3\},\{1,2\}}, \mathcal{P}_{\aleph,\{1,3\},\{2\}}, \mathcal{P}_{\aleph, \{3\}, \{1,2\}},  \mathcal{P}_{\aleph, \{3\},\{2\}}  \}.
\end{align}
\end{small}

Now, the sets $\mathscr{Q}_{\mathcal{T}}$'s have this property that for any $\mathcal{P} \in \mathscr{Q}_{\mathcal{T}}$, we have $|\mathcal{P}| = |\mathcal{T}|$. Remember that in the function \textsc{PBO}, we want to find a feasible packet $\mathcal{P} \in \mathscr{S}$ \black{that maximizes} $\frac{|\mathcal{P}|}{\bc{\mathcal{P}}}$. Using the decomposition of space $\mathscr{S}$ described above, instead of searching over space $\mathscr{S}$, we can \black{search} over the 7 subsets $\mathscr{Q}_{\mathcal{T}}$ separately. In other words, 
\begin{align}
\bc{\textsc{PBO}(\mathscr{S})} = \max_{\mathcal{P} \in \mathscr{S}} \frac{|\mathcal{P}|}{\bc{\mathcal{P}}} = \max_{\mathcal{T}}\max_{\mathcal{\mathcal{P}} \in \mathscr{Q}_{\mathcal{T}}}\frac{|\mathcal{P}|}{\bc{\mathcal{P}}}.
\label{eq:PBO_to_SBO_1}
\end{align}
Since $|\mathcal{P}| = |\mathcal{T}|$ for all $\mathcal{P} \in \mathscr{Q}_{\mathcal{T}}$, maximizing $\frac{|\mathcal{P}|}{\bc{\mathcal{P}}}$ over $\mathcal{P} \in\mathscr{Q}_{\mathcal{T}}$  is equivalent to minimizing $\bc{\mathcal{P}}$ over $\mathcal{P} \in \mathscr{Q}_{\mathcal{T}}$.

Now, \black{let us} focus on minimizing $\bc{\mathcal{P}}$ over $\mathcal{P} \in \mathscr{Q}_{\mathcal{T}}$ for a specific set $\mathcal{T}$. 
Taking for example $\mathcal{T} = \{1,2\}$, we can define the sets $\mathcal{L}_{1,\{1,2\}}$ and $\mathcal{L}_{2,\{1,2\}}$ of subfiles as in Fig. \ref{fig:T_1_2}, that is, $\mathcal{L}_{1,\{1,2\}}:= \{W_{1,\{2\}}, W_{1,\{2,3\}}\}$ and $\mathcal{L}_{2,\{1,2\}}:= \{W_{2,\{1\}}, W_{2,\{1,3\}}\}$. Then any feasible packet $\mathcal{P} \in \mathscr{Q}_{\{1,2\}}$, being a clique in Fig. \ref{fig:T_1_2}, has exactly one subfile from $\mathcal{L}_{1,\{1,2\}}$ and one subfile from $\mathcal{L}_{2,\{1,2\}}$. 
Hence, minimizing $\bc{\mathcal{P}}$ over all $\mathcal{P} \in \mathscr{Q}_{\{1,2\}}$ is the same as finding the subfile\black{s} with \black{the minimum size} in both $\mathcal{L}_{1,\{1,2\}}$ and $\mathcal{L}_{2,\{1,2\}}$.
We use $V_{1,\{1,2\}}$ and $V_{2,\{1,2\}}$ to denote the two subfiles with the minimum sizes in $\mathcal{L}_{1,\{1,2\}}$ and $\mathcal{L}_{2,\{1,2\}}$, respectively.
As it has been shown in Fig. \ref{fig:T_1_2} with dark color, we have $V_{1,\{1,2\}} = W_{1,\{2\}}$ and $V_{2,\{1,2\}} = W_{2,\{1,3\}}$. Then $\mathcal{R}_{\{1,2\}} = \{V_{1,\{1,2\}},V_{2,\{1,2\}}\} = \{W_{1,\{2\}}, W_{2,\{1,3\}}\}$ is a clique belonging to $\mathscr{Q}_{\{1,2\}}$ and it has the minimum size. 
Similarly, for all $\mathcal{T} \in \mathscr{T}$, we can find $\mathcal{R}_{\mathcal{T}}$ such that  $\bc{\mathcal{R}_{\mathcal{T}}} = \min_{\mathcal{P} \in \mathscr{Q}_{\mathcal{T}}}\bc{\mathcal{P}}$, and consequently, $\max_{\mathcal{P} \in \mathscr{Q}_{\mathcal{T}}}\frac{|\mathcal{P}|}{\bc{\mathcal{P}}} = \frac{|\mathcal{T}|}{\bc{\mathcal{R}_{\mathcal{T}}}}$. This together with \eqref{eq:PBO_to_SBO_1} \black{implies} that
\begin{align}
\bc{\textsc{PBO}(\mathscr{S})} = \max_{\mathcal{T} \in \mathscr{T}}\max_{\mathcal{\mathcal{P}} \in \mathscr{Q}_{\mathcal{T}}}\frac{|\mathcal{P}|}{\bc{\mathcal{P}}}
= \max_{\mathcal{T} \in \mathscr{T}}  \frac{|\mathcal{T}|}{\bc{\mathcal{R}_{\mathcal{T}}}}.
\label{eq:PBO_to_SBO_2}
\end{align}
Now, \eqref{eq:PBO_to_SBO_2} provides another way of calculating $\bc{\textsc{PBO}(\mathscr{S})}$ that does not involve the set $\mathscr{S}$ of cliques\black{;} instead all it needs is the collection of sets $\mathcal{R}_{\mathcal{T}}$ for $\mathcal{T} \in \mathscr{T}$. In order to find $\mathcal{R}_{\mathcal{T}}$ for any $\mathcal{T} \in \mathscr{T}$, all we need is to find the sets $\mathcal{L}_{j,\mathcal{T}}$ for $j \in \mathcal{T}$, which can be done only by using the set $\mathcal{E}$ (See definition of $\mathcal{L}_{j,\mathcal{T}}$ in \black{the} function $\textsc{SBO}(\mathcal{E})$). Since $\max_{\mathcal{T} \in \mathscr{T}}  \frac{|\mathcal{T}|}{\bc{\mathcal{R}_{\mathcal{T}}}}$ is the size of \black{the} output clique of \black{the} function $\textsc{SBO}(\mathcal{E})$, from \eqref{eq:PBO_to_SBO_2} we have, $\bc{\textsc{PBO}(\mathscr{S})} = \bc{\textsc{SBO}(\mathcal{E})}$.

\item How to calculate $\textsc{SBO}(\mathcal{E})$? To this end, we first find sets $\mathcal{M}$ and $\mathscr{T}$ as described in \black{the} function $\textsc{SBO}(\mathcal{E})$. Note that $\mathcal{M}$ is the set of users' indices for which there exists a subfile in $\mathcal{E}$ and $\mathscr{T}$ is the set of all non-empty subsets of $\mathcal{M}$. Then, for each $\mathcal{T} \in \mathscr{T}$, we find sets $\mathcal{L}_{j,\mathcal{T}}$ for all $j \in \mathcal{T}$. Sets $\mathcal{L}_{j,\mathcal{T}}$ are denoted by blue rectangles in Fig. \ref{fig:proof}. For each set $\mathcal{L}_{j,\mathcal{T}}$, $V_{j,\mathcal{T}}$, which is the subfile with the smallest size, is denoted in dark color. Furthermore, for each $\mathcal{T} \in \mathscr{T}$, set $\mathcal{R}_{\mathcal{T}}$ which includes the subfiles $V_{j,\mathcal{T}}$ for all $j \in \mathcal{T}$ is denoted by a red ellipse. Then, the output of function $\textsc{SBO}(\mathcal{E})$ is the clique $\mathcal{R}_{\mathcal{T}}$ maximizing $\frac{|\mathcal{T}|}{\bc{\mathcal{R}_{\mathcal{T}} }}$.
\end{itemize}

\end{examp}

}

\begin{small}
\begin{figure}
\begin{subfigure}[b]{0.23\textwidth}
  \begin{center}
\begin{tikzpicture}[minimum size=10mm, node distance=1cm, every loop/.style={},
                    thick,state/.style={circle, inner sep=0pt, draw,font=\sffamily\small\bfseries}]
\begin{scope} [rotate=-90, scale=0.5, every node/.append style={scale=0.5}]
\node[state, fill=black!14] (A_1) at (-4,0) {$W_{1,\{2,3\}}$};
\draw[blue,thick, line width=0.4mm,dotted,label=right:$m_1$]     ($(A_1.north)+(-0.2,0.65)$) rectangle ($(A_1.south)+(0.2,-0.65)$);
\node[left=2mm of A_1] (name) {\bblue{\Large $\mathcal{L}_{1,\{1,2, 3\}}$}};

\node[state,fill=black!14] (A_2) at (-2,0) {$W_{2,\{1,3\}}$};
\draw[blue,thick, line width=0.4mm,dotted,label=right:$m_1$]     ($(A_2.north)+(-0.2,0.65)$) rectangle ($(A_2.south)+(0.2,-0.65)$);
\node[left=2mm of A_2] (name) {\bblue{\Large $\mathcal{L}_{2,\{1,2, 3\}}$}};

\node[state,fill=black!14] (A_3) at (0,0) {$W_{3,\{1,2\}}$};
\draw[blue,thick, line width=0.4mm,dotted, label=right:$m_1$]     ($(A_3.north)+(-0.2,0.65)$) rectangle ($(A_3.south)+(0.2,-0.65)$);
\node[left=2mm of A_3] (name) {\bblue{\Large $\mathcal{L}_{3, \{1,2, 3\}}$}};


\draw[red,thick, line width=0.3mm, rotate=0] (-2,0) ellipse (3.2cm and 0.95cm);
\node[right=2mm of A_2] (name) {\red{\Large $\mathcal{R}_{\{1,2, 3\}}$}};

\draw[every node/.style={font=\sffamily\small}]
(A_1)  edge node {}   (A_2)

(A_2)  edge node {}   (A_3)

 (A_1) edge [bend left] node {}   (A_3);
\end{scope}
\end{tikzpicture}
\end{center}
\caption{$\mathcal{T}=\{1,2,3\}$}
\label{fig:T_1_2_3}
\end{subfigure}
\begin{subfigure}[b]{0.23\textwidth}
\begin{tikzpicture}[minimum size=10mm, node distance=1cm, every loop/.style={},
                    thick,state/.style={circle, inner sep=0pt, draw,font=\sffamily\small\bfseries}]
\begin{scope} [rotate=-90, scale=0.5, every node/.append style={scale=0.5}]
\node[state] (A_1) at (-4,0) {$W_{1,\{2,3\}}$};
\node[state,fill=black!14] (C_1) at (-4,-4) {$W_{1,\{3\}}$};
\draw[blue,thick, line width=0.4mm,dotted,label=right:$m_1$]     ($(A_1.north)+(-0.2,0.65)$) rectangle ($(C_1.south)+(0.2,-0.65)$);
\node[above=2mm of A_1] (name) {\bblue{\Large $\mathcal{L}_{1, \{1, 3\}}$}};

\node[state] (A_3) at (0,0) {$W_{3,\{1,2\}}$};
\node[state,fill=black!14] (B_3) at (0,-2) {$W_{3,\{1\}}$};
\draw[blue,thick, line width=0.4mm,dotted,label=right:$m_1$]     ($(A_3.north)+(-0.2,0.65)$) rectangle ($(B_3.south)+(0.2,-0.65)$);
\node[below=2mm of A_3] (name) {\bblue{\Large $\mathcal{L}_{3, \{1,3\}}$}};

\draw[red,thick, line width=0.3mm, rotate=25] (-3,-1.9) ellipse (3.6cm and 0.95cm);
\node (name) at (-2,-5) {\red{\Large $\mathcal{R}_{\{1,3\}}$}};

\draw[every node/.style={font=\sffamily\small}]
(A_1)  edge node {}   (B_3)

 (A_1) edge [bend left] node {}   (A_3)
(A_3)  edge [bend right =20] node {}   (C_1)

(C_1)  edge node {}   (B_3);
\end{scope}
\end{tikzpicture}
\caption{$\mathcal{T}=\{1,3\}$}
\label{fig:T_1_3}
\end{subfigure}
\begin{subfigure}[b]{0.23\textwidth}
  \begin{center}
\begin{tikzpicture}[minimum size=10mm, node distance=1cm, every loop/.style={},
                    thick,state/.style={circle, inner sep=0pt, draw,font=\sffamily\small\bfseries}]
\begin{scope} [rotate=-90, scale=0.5, every node/.append style={scale=0.5}]

\node[state] (A_2) at (-2,0) {$W_{2,\{1,3\}}$};
\node[state,fill=black!14] (C_2) at (-2,-4) {$W_{2,\{3\}}$};
\draw[blue,thick, line width=0.4mm,dotted,label=right:$m_1$]     ($(A_2.north)+(-0.2,0.65)$) rectangle ($(C_2.south)+(0.2,-0.65)$);
\node[above=1mm of A_2] (name) {\bblue{\Large $\mathcal{L}_{2, \{2, 3\}}$}};

\node[state] (A_3) at (0,0) {$W_{3,\{1,2\}}$};
\node[state,fill=black!14] (C_3) at (0,-4) {$W_{3,\{2\}}$};
\draw[blue,thick, line width=0.4mm,dotted,label=right:$m_1$]     ($(A_3.north)+(-0.2,0.65)$) rectangle ($(C_3.south)+(0.2,-0.65)$);
\node[below=1mm of A_3] (name) {\bblue{\Large $\mathcal{L}_{3, \{2, 3\}}$}};

\draw[red,thick, line width=0.3mm, rotate=0] (-1,-4) ellipse (2.0cm and 0.85cm);

\node[above=2mm of C_2] (name)  {\red{\Large $\mathcal{R}_{\{2,3\}}$}};

\draw[every node/.style={font=\sffamily\small}]

(A_2)  edge node {}   (A_3)
(A_2)  edge node {}   (C_3)

(A_3)  edge  node {}   (C_2)

(C_2)  edge node {}   (C_3);
\end{scope}
\end{tikzpicture}
\end{center}
\caption{$\mathcal{T}=\{2,3\}$}
\label{fig:T_2_3}
\end{subfigure}
\begin{subfigure}[b]{0.23\textwidth}
  \begin{center}
\begin{tikzpicture}[minimum size=10mm, node distance=1cm, every loop/.style={},
                    thick,state/.style={circle, inner sep=0pt, draw,font=\sffamily\small\bfseries}]
\begin{scope} [rotate=-90, scale=0.5, every node/.append style={scale=0.5}]
\node[state] (A_1) at (-4,0) {$W_{1,\{2,3\}}$};
\node[state,fill=black!14] (B_1) at (-4,-2) {$W_{1,\{2\}}$};
\draw[blue,thick, line width=0.4mm,dotted,label=right:$m_1$]     ($(A_1.north)+(-0.2,0.65)$) rectangle ($(B_1.south)+(0.2,-0.65)$);
\node[left=2mm of B_1] (name) {\bblue{\Large $\mathcal{L}_{1, \{1,2\}}$}};

\node[state,fill=black!14] (A_2) at (-2,0) {$W_{2,\{1,3\}}$};
\node[state] (B_2) at (-2,-2) {$W_{2,\{1\}}$};
\draw[blue,thick, line width=0.4mm,dotted,label=right:$m_1$]     ($(A_2.north)+(-0.2,0.65)$) rectangle ($(B_2.south)+(0.2,-0.65)$);
\node[left=2mm of B_2] (name) {\bblue{\Large $\mathcal{L}_{2, \{1,2\}}$}};

\draw[red,thick, line width=0.3mm, rotate=45] (-3,1.4) ellipse (2.8cm and 0.8cm);
\node[below=1mm of A_2] (name) {\red{\Large $\mathcal{R}_{\{1,2\}}$}};

\draw[every node/.style={font=\sffamily\small}]
(A_1)  edge node {}   (A_2)
(A_1)  edge node {}   (B_2)

(A_2)  edge node {}   (B_1)

(B_1)  edge node {}   (B_2);
\end{scope}
\end{tikzpicture}
\end{center}
\caption{$\mathcal{T}=\{1,2\}$}
\label{fig:T_1_2}
\end{subfigure}
\begin{subfigure}[b]{0.25\textwidth}
  \begin{center}
\begin{tikzpicture}[minimum size=10mm, node distance=1cm, every loop/.style={},
                    thick,state/.style={circle, inner sep=0pt, draw,font=\sffamily\small\bfseries}]
\begin{scope} [rotate=-90, scale=0.5, every node/.append style={scale=0.5}]
\node[state] (A_3) at (0,0) {$W_{3,\{1,2\}}$};
\node[state] (B_3) at (0,-2) {$W_{3,\{1\}}$};
\node[state] (C_3) at (0,-4) {$W_{3,\{2\}}$};
\node[state,fill=black!14] (D_3) at (0,-6) {$W_{3,\varnothing}$};
\draw[blue,thick, line width=0.4mm,dotted,label=right:$m_1$]     ($(A_3.north)+(-0.2,0.65)$) rectangle ($(D_3.south)+(0.2,-0.65)$);
\node[above=1mm of B_3] (name) {\bblue{\Large $\mathcal{L}_{3, \{3\}}$}};

\draw[red,thick, line width=0.3mm, rotate=0] (0,-6) ellipse (1cm and 0.6cm);
\node (name) at (-1.5,-6) {\red{\Large $\mathcal{R}_{\{3\}}$}};

\end{scope}
\end{tikzpicture}
\end{center}
\caption{$\mathcal{T}=\{3\}$}
\label{fig:T_3}
\end{subfigure}
\begin{subfigure}[b]{0.25\textwidth}
  \begin{center}
\begin{tikzpicture}[minimum size=10mm, node distance=1cm, every loop/.style={},
                    thick,state/.style={circle, inner sep=0pt, draw,font=\sffamily\small\bfseries}]
\begin{scope} [rotate=-90, scale=0.5, every node/.append style={scale=0.5}]
\node[state] (A_2) at (-2,0) {$W_{2,\{1,3\}}$};
\node[state] (B_2) at (-2,-2) {$W_{2,\{1\}}$};
\node[state] (C_2) at (-2,-4) {$W_{2,\{3\}}$};
\node[state,fill=black!14] (D_2) at (-2,-6) {$W_{2,\varnothing}$};
\draw[blue,thick, line width=0.4mm,dotted,label=right:$m_1$]     ($(A_2.north)+(-0.2,0.65)$) rectangle ($(D_2.south)+(0.2,-0.65)$);
\node[above=1mm of B_2] (name) {\bblue{\Large $\mathcal{L}_{2, \{2\}}$}};

\draw[red,thick, line width=0.3mm, rotate=0] (-2,-6) ellipse (1cm and 0.6cm);
\node (name) at (-3.5,-6) {\red{\Large $\mathcal{R}_{\{2\}}$}};

\end{scope}
\end{tikzpicture}
\end{center}
\caption{$\mathcal{T}=\{2\}$}
\label{fig:T_2}
\end{subfigure}
\begin{subfigure}[b]{0.5\textwidth}
  \begin{center}
\begin{tikzpicture}[minimum size=10mm, node distance=1cm, every loop/.style={},
                    thick,state/.style={circle, inner sep=0pt, draw,font=\sffamily\small\bfseries}]
\begin{scope} [rotate=-90, scale=0.5, every node/.append style={scale=0.5}]
\node[state] (A_1) at (-4,0) {$W_{1,\{2,3\}}$};
\node[state] (B_1) at (-4,-2) {$W_{1,\{2\}}$};
\node[state] (C_1) at (-4,-4) {$W_{1,\{3\}}$};
\node[state,fill=black!14] (D_1) at (-4,-6) {$W_{1,\varnothing}$};
\draw[blue,thick, line width=0.4mm,dotted,label=right:$m_1$]     ($(A_1.north)+(-0.2,0.65)$) rectangle ($(D_1.south)+(0.2,-0.65)$);
\node[above=1mm of B_1] (name) {\bblue{\Large $\mathcal{L}_{1, \{1\}}$}};

\draw[red,thick, line width=0.3mm, rotate=0] (-4,-6) ellipse (1cm and 0.6cm);
\node (name) at (-5.5,-6) {\red{\Large $\mathcal{R}_{\{1\}}$}};

\end{scope}
\end{tikzpicture}
\end{center}
\caption{$\mathcal{T}=\{1\}$}
\label{fig:T_1}
\end{subfigure}%
\caption{(\ref{fig:T_1_2_3})-(\ref{fig:T_1}): Graphical representation of sets $\mathcal{L}_{j,\mathcal{T}}$ for the first iteration of Algorithm \ref{my_delivery} applied to Example \ref{first_example}. Sets $\mathcal{L}_{j,\mathcal{T}}$ are denoted by blue rectangles. For each set $\mathcal{L}_{j,\mathcal{T}}$, the corresponding subfile $V_{j,\mathcal{T}}$ is denoted in dark color. Sets $\mathcal{R}_{\mathcal{T}}$ are denoted by red ellipses.}
\label{fig:proof}
\end{figure}
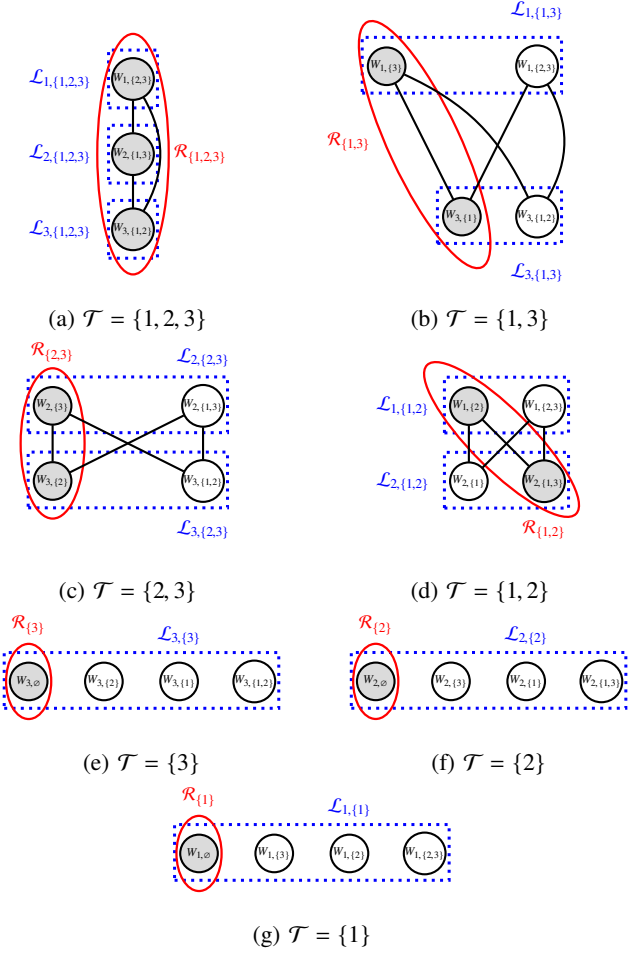
\end{small}

Note that according to Theorem \ref{equivalent_function}, the function \textsc{PBO} in Algorithm \ref{first_delivery} can be replaced by the function \textsc{SBO} without altering the optimality bound for the algorithm.
By doing so, 
we can also eliminate the steps of updating set $\mathscr{S}$ (lines 8-10 of Algorithm \ref{first_delivery}) since we do not need $\mathscr{S}$ as input to the function \textsc{SBO}.
The resulting algorithm is described in Algorithm \ref{my_delivery}. The optimality bound 
of Algorithm \ref{my_delivery} is stated in the following theorem.
\begin{thm}
\label{optimal_delivery_solution_second}
Algorithm \ref{my_delivery} \black{achieves a $(1 + \log K)$-approximation} to Problem \ref{optimal_delivery_N}.
\end{thm}
\begin{proof}
See Appendix \ref{proof:optimal_delivery}.
\end{proof}
\begin{small}
\begin{algorithm}[h]
\caption{SACM (Size-Aware Coded Multicast)}
\label{my_delivery}
\hspace*{\algorithmicindent} {\small \textbf{Input:} Set of subfiles $\mathcal{W}$.}
\begin{algorithmic}[1]
\State $\mathscr{C} = \emptyset$;
\State $\mathcal{E} =\mathcal{W}$;
\While{$\mathcal{E} \neq \emptyset$}
\State $\mathcal{P^{\black{*}}} =$ \textsc{SBO}$(\mathcal{E})$;
\State $\mathscr{C} = \mathscr{C} \cup \{\mathcal{P^{\black{*}}}\}$;
\State $\mathcal{E} = \mathcal{E} \setminus \mathcal{P^{\black{*}}}$;
\EndWhile
\end{algorithmic}
\hspace*{\algorithmicindent} {\small \textbf{Output:} Set of cliques $\mathscr{C}$.}
\end{algorithm} 
\end{small}
If we apply Algorithm \ref{my_delivery} to Example \ref{first_example}, the \black{cliques} $\mathcal{P}_{\aleph, \{3\},\{2\}}$, $\mathcal{P}_{\{2\},\{1,3\}, \aleph}$, $\mathcal{P}_{\{3\},\aleph, \{1\}}$, $\mathcal{P}_{\varnothing, \aleph, \aleph}$, $\mathcal{P}_{\aleph, \varnothing, \aleph}$, $\mathcal{P}_{\aleph, \aleph,\varnothing}$,
$\mathcal{P}_{\aleph, \aleph,\{1,2\}}$, and $\mathcal{P}_{\{2,3\},\{1\}, \aleph}$ are chosen with the total size of $7 \times 10 +300 = 370$. Fig. \ref{fig:proof} illustrates a graphical representation of sets $\mathcal{L}_{j,\mathcal{T}}$, subfiles $V_{j,\mathcal{T}}$, and sets $\mathcal{R}_{\mathcal{T}}$ for the first iteration of Algorithm \ref{my_delivery} applied to Example \ref{first_example}. Compared to \black{the} total number of bits for \black{the} optimal solution to Example \ref{first_example}, we can see that for this example, the total number of bits for our algorithm is equal to that of \black{the} optimal solution.

{\color{black}
\subsection{Comparison to State-of-the-Art}
\label{sec:approximation_ratio}
We theoretically compare the worst\black{-}case approximation ratio between our proposed algorithm with the uncoded delivery, the GCM scheme \cite{maddah2014fundamental,maddah2015decentralized,niesen2017coded,Zhang_diff_sizes_2015}, \black{GCC scheme \cite{ji2014order}, GCLC and HgLC schemes \cite{ji2015efficient}}, and the GCCM scheme \cite{hueristics_index_coding}\footnote{We do not compare with the proposed schemes of \cite{Ramakrishnan_2015, Wan_2017} as these algorithms take advantage of the \black{sub-packetization} which is not allowed in the problem we study in this paper.}. While Theorem \ref{optimal_delivery_solution_second} states that the approximation ratio for our proposed algorithm is $(1 + \log K)$, the following result show\black{s} that the ratio for the other mentioned schemes is linear in $K$.

\begin{thm}
\label{thm:perfromance_SID}
For a system with $K$ users, there are instances on which the approximation ratio~of
\begin{enumerate}
\item the uncoded delivery is only $K$,
\item the GCM (greedy coded multicast) scheme (and similarly \black{GCC, GCLC, and HgLC} schemes) is only $\lfloor \frac{K-1}{2}\rfloor $,
\item the GCCM (graph coloring-based coded multicast) scheme is only $K-1$.
\end{enumerate}
\end{thm}
\begin{proof}
See Appendix \ref{proof:thm:perfromance_SID}.
\end{proof}
According to Theorems \ref{optimal_delivery_solution_second} and \ref{thm:perfromance_SID}, our proposed algorithm (Algorithm \ref{my_delivery}) provides a significantly better approximation ratio compared to the above schemes. 

\section{Complexity Analysis}
\label{clique_cover_description}

In this section, we first analyze the complexity of SACM algorithm we proposed in Section \ref{sec:approximation}.
It can be seen that SACM needs $O(\tau_K K2^{K})$ operations where $\tau_K = |\mathcal W|$ is the number of non-empty subfiles and it can change\footnote{It is $K2^{K-1}$ when for each $k \in [K]$ we have $\bc{W_{k,\mathcal{A}}} \neq 0$ for all $\mathcal{A} \subseteq [K] \setminus \{k\}$.} from $K$ to $K2^{K-1}$. The GCM scheme \cite{maddah2014fundamental,maddah2015decentralized,niesen2017coded,Zhang_diff_sizes_2015} needs $O(K2^{K})$ operations while \black{the GCC, GCLC, and HgLC schemes} as well as the GCCM scheme \cite{hueristics_index_coding} require $O(\tau^2_K)$. \black{Note that in most existing placement schemes, e.g., \cite{maddah2015decentralized,niesen2017coded,ji2014order}, almost all subfiles are non-empty, that is, the number of subfiles is exponential in $K$ (more precisely, $\tau_K \approx K2^{K-1}$). This means that under these placement schemes, SACM, GCM, \black{GCC, GCLC, HgLC, }and GCCM all have the same complexity which is exponential with respect to the number $K$ of users. If the number of subfiles is not exponential in $K$ \cite{maddah2014fundamental}, then \black{GCC, GCLC, HgLC, }and GCCM schemes achieve lower complexity compared to SACM and GCM schemes.}

Next, we show that the exponential complexity (in \black{the} number $K$ of users) is inevitable for any algorithm with good approximation ratio for Problem \ref{optimal_delivery_N}. This result which is presented in the following theorem is based on the facts that minimum clique cover problem is a special case of Problem \ref{optimal_delivery_N} and it is NP-hard to find a polynomial\black{-}time approximation for the minimum clique cover problem \cite{zuckerman2006linear}.

\begin{thm}
\label{hardness_approximation}
Unless P = NP, there is no polynomial\black{-}time (in \black{the} number $K$ of users) algorithm for Problem \ref{optimal_delivery_N} with the approximation ratio of $K^{1-\varepsilon}$ for any $\varepsilon>0$.
\end{thm}
\begin{proof}
See Appendix \ref{proof:hardness_approximation}.
\end{proof}

\section{Numerical Experiments}
\label{sec:num_exp}
In this section, we compare our proposed algorithm, that is, Size-Aware Coded Multicast (SACM) with the GCM
(greedy coded multicast) scheme \cite{maddah2014fundamental,maddah2015decentralized,niesen2017coded,Zhang_diff_sizes_2015}, 
\black{GCC scheme \cite{ji2014order}, GCLC and HgLC schemes \cite{ji2015efficient}}, the GCCM scheme \cite{hueristics_index_coding}, and \black{the} uncoded delivery. For a system with $K$ users, let $\tau_K$ denote the number of subfiles with non-zero size required to be sent by the server, that is, $\tau_K =  |\mathcal{W}|$. We assume that the size of each subfile in the set $\mathcal{W}$ can be between $1$ to $1000$ bits. In order to make sure that our comparison does not depend on a particular content placement, we take $100$ samples uniformly among all placement possibilities. We consider a system with the number of users $K=3,6,8,10$. For each system, the number $\tau_K$ of non-zero subfiles can change from $K$ to $K2^{K-1}$.
For $K=3$, since the number $\tau_3$ of subfile is at most 12, it is possible to find all feasible packets and solve the ILP of Problem \ref{ILP_problem}. This is the reason the optimal solution of Problem \ref{optimal_delivery_N} is only calculated for $K=3$. 
\begin{figure}
\centering
\begin{subfigure}{0.47\textwidth}
 \begin{center}
\includegraphics[width=8.8cm,height=4.5cm]{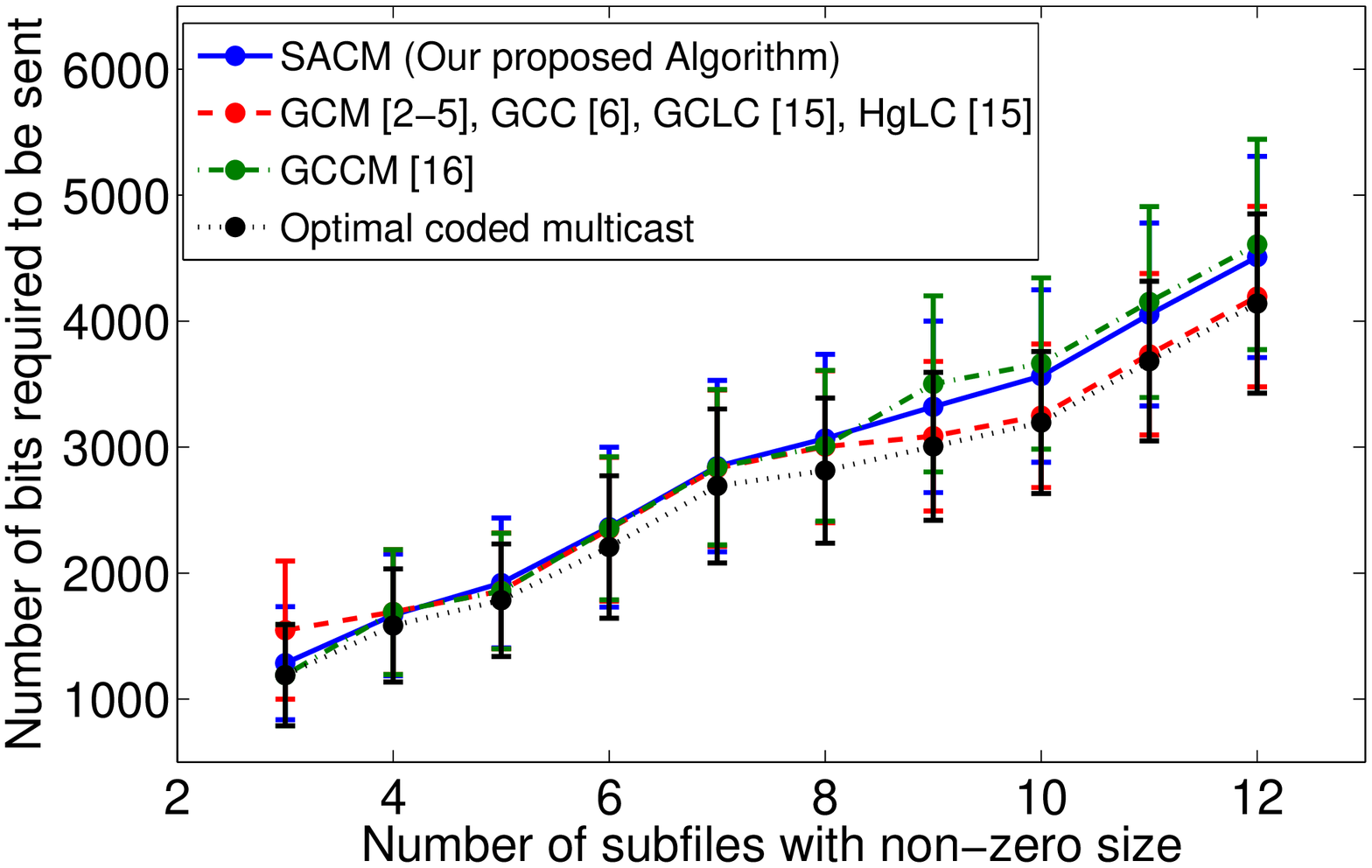}
\end{center}
\caption{{\footnotesize Comparison of our proposed algorithm (SACM), GCM, \black{GCC, GCLC, HgLC,} GCCM, and the optimal solution of Problem 1 for a system of $K=3$ users.}}
\label{fig:N3_err}
\end{subfigure}
\begin{subfigure}{0.47\textwidth}
 \begin{center}
\includegraphics[width=8.8cm,height=4.5cm]{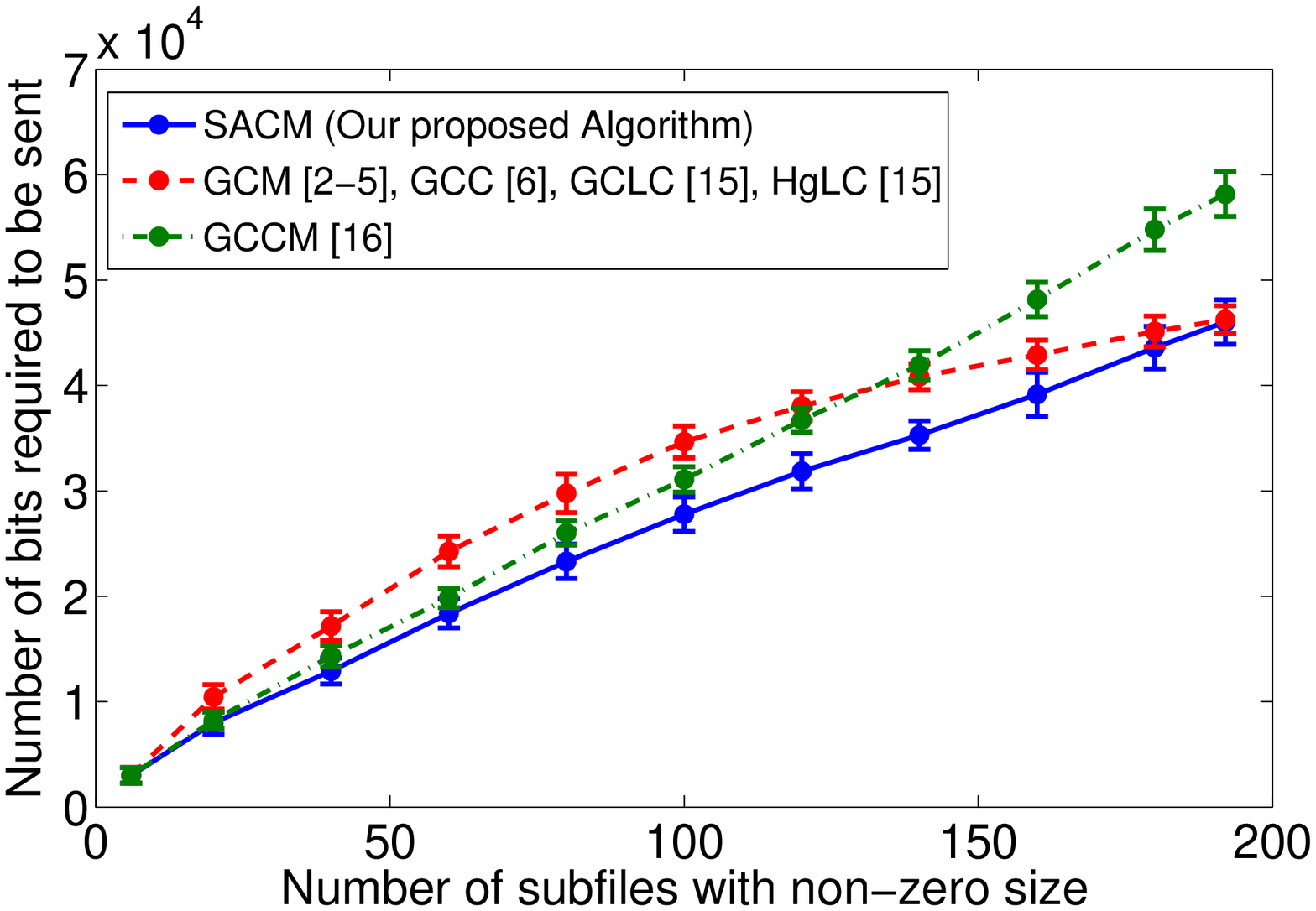}
\end{center}
\caption{{\footnotesize Comparison of our proposed algorithm (SACM), \black{GCC, GCLC, HgLC,}  and GCCM for a system of $K=6$ users.}}
\label{fig:N6_err}
\end{subfigure}
\begin{subfigure}{0.47\textwidth}
\begin{center}
\includegraphics[width=8.8cm,height=4.5cm]{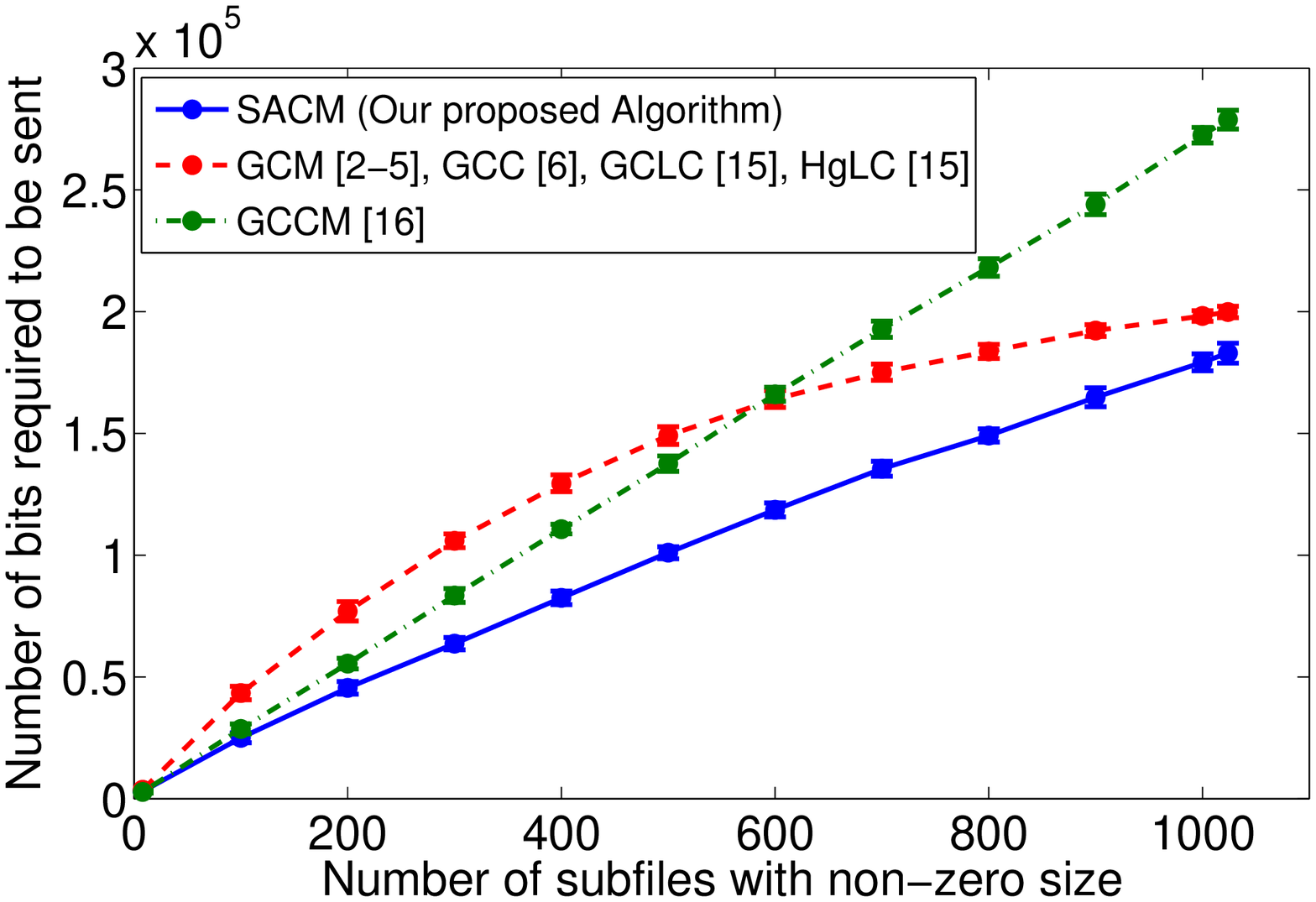}
\end{center}
\caption{{\footnotesize Comparison of our proposed algorithm (SACM), GCM, \black{GCC, GCLC, HgLC,} and GCCM for a system of $K=8$ users.}}
\label{fig:N8_err}
\end{subfigure}
\begin{subfigure}{0.47\textwidth}
 \begin{center}
\includegraphics[width=8.8cm,height=4.5cm]{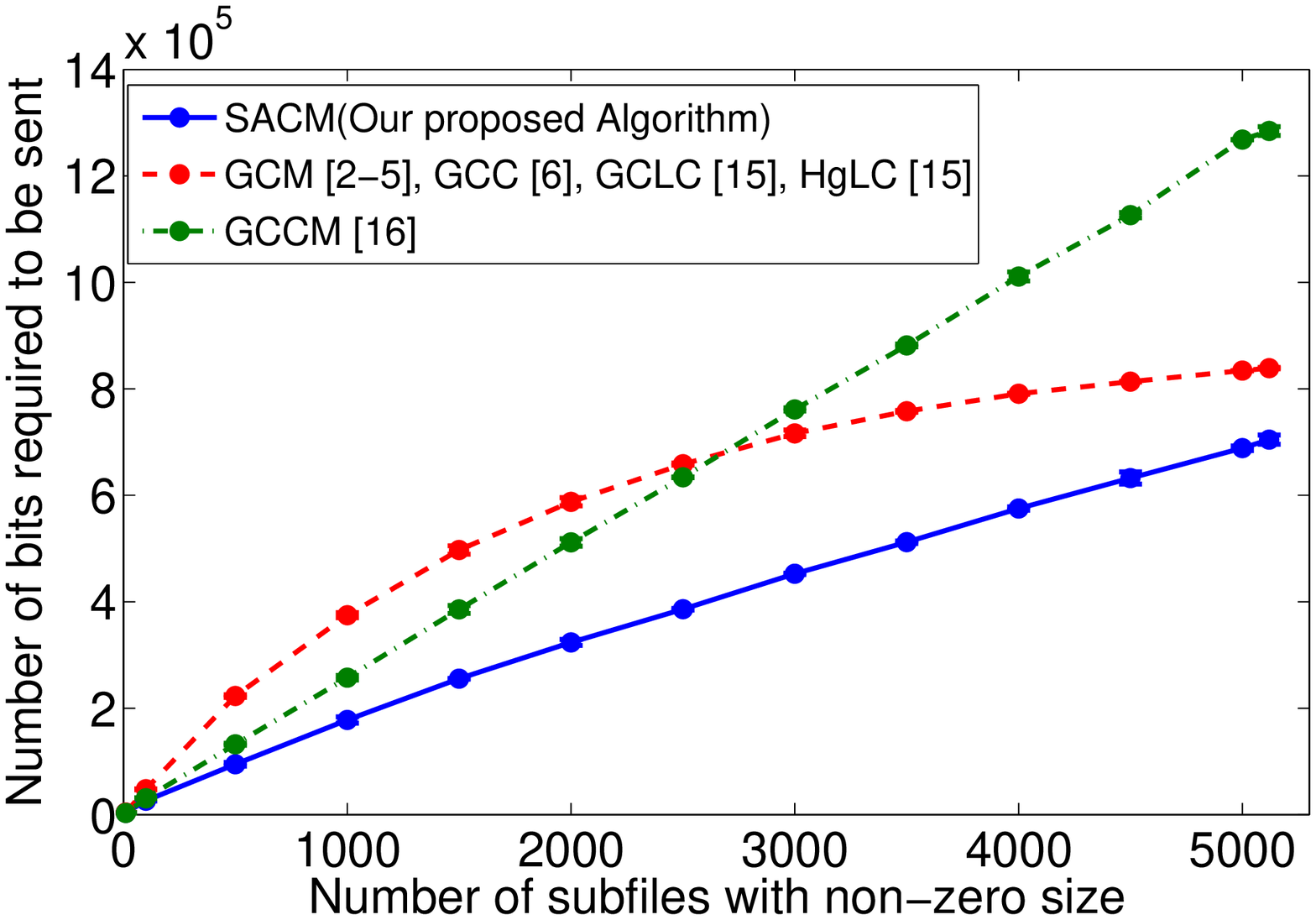}
\end{center}
\caption{{\footnotesize Comparison of our proposed algorithm (SACM), \black{GCC, GCLC, HgLC,}  and GCCM for a system of $K=10$ users.}}
\label{fig:N10_err}
\end{subfigure}
\caption{}
\end{figure}
Fig. \ref{fig:N3_err}-\ref{fig:N6_err} shows the average number of bits requires to be sent by the \black{server} versus the number of subfiles with non-zero size for $K=3,6$. As can be seen from Fig. \ref{fig:N3_err}, when $K=3$, all coded delivery schemes (that is, SACM, GCM, \black{GCC, GCLC, HgLC,} and GCCM schemes) perform almost the same as the optimal \black{clique cover delivery scheme} as their confidence intervals overlap. However, as can be observed from Fig. \ref{fig:N6_err}, when $K=6$, our proposed algorithm (SACM) outperforms GCM, \black{GCC, GCLC, HgLC,} and GCCM. Further, it can be observed that as the number of subfiles with non-zero size increases, the ratio of \black{the} number of bits required to be sent by GCCM scheme to that of our proposed algorithm (SACM) increases while the ratio of GCM, \black{GCC, GCLC, and HgLC schemes} to our algorithm (SACM) increases at first and then decreases.

In order to study the effect of number $K$ of users, we have plotted the average number of bits requires to be sent by the \black{server} versus the number of subfiles with non-zero size for $K=8,10$. As can be seen from Fig. \ref{fig:N3_err}-\ref{fig:N10_err}, our proposed algorithm (SACM) significantly outperforms GCM, \black{GCC, GCLC, HgLC,} and GCCM schemes when there are more users. Our simulation results indicate that for a system of 10 users, our proposed algorithm (SACM) can reduce the number of bits required to be sent between 15\% to 45\% compared to GCCM scheme and between 16\% to 57\% compared to GCM, GCC, GCLC, and HgLC schemes.
\begin{small}
\begin{table}[h]
\caption{The bandwidth reduction for our proposed algorithm compared to uncoded delivery. For each number $K$ of users, the range indicates the lower and upper bounds of the bandwidth reduction for \black{a} various number of subfiles with non-zero sizes. }
\begin{center}
\begin{tabular}{l || c | c | c | c}
Number of users ($K$)             & 3 & 6 & 8 & 10  \\
\hline
Bandwidth Reduction \% & 17 to 24 & 25 to 60 & 43 to 62 & 46 to 72
\end{tabular}
\end{center}
\label{ours_uncoded_comprison}
\end{table}
\end{small}

Table \ref{ours_uncoded_comprison} indicates the bandwidth reduction for our proposed algorithm compared to \black{the} uncoded delivery. As can be seen from the table, for a system of 10 users, our proposed algorithm can reduce the number of bits required to be sent up to 72\%.
\section{Conclusion}
\label{sec:conc}
In this work, given an arbitrary content placement, we formulated the optimal \black{clique cover delivery problem} and showed that it can be represented as an Integer Linear Program (ILP). We then proposed an approximation algorithm for the optimal \black{clique cover delivery problem} by investigating the connection between our problem and the weighted set cover problem. We showed that for a system with $K$ users, our proposed algorithm provides $(1 + \log K)$-approximation for the optimal \black{clique cover delivery problem}, while the approximation ratio for the existing coded delivery schemes is linear in $K$. Further\black{,} through simulations, we showed that our proposed algorithm can remarkably decrease the bandwidth required for satisfying the demands of the users compared to the existing coded delivery schemes for almost all content placement schemes.

\bibliographystyle{ieeetr}
\bibliography{caching_ref}

\begin{thebibliography}{10}

\bibitem{Asghari_ICC_2018}
S.~M. Asghari, Y.~Ouyang, A.~Nayyar, and A.~S. Avestimehr, ``Optimal coded
  multicast in cache networks with arbitrary content placement,'' in {\em IEEE
  International Conference on Communications (ICC)}, May 2018.

\bibitem{maddah2014fundamental}
M.~A. Maddah-Ali and U.~Niesen, ``Fundamental limits of caching,'' {\em IEEE
  Transactions on Information Theory}, vol.~60, no.~5, pp.~2856--2867, 2014.

\bibitem{maddah2015decentralized}
M.~A. Maddah-Ali and U.~Niesen, ``Decentralized coded caching attains
  order-optimal memory-rate tradeoff,'' {\em IEEE/ACM Transactions On
  Networking}, vol.~23, no.~4, pp.~1029--1040, 2015.

\bibitem{niesen2017coded}
U.~Niesen and M.~A. Maddah-Ali, ``Coded caching with nonuniform demands,'' {\em
  IEEE Transactions on Information Theory}, vol.~63, no.~2, pp.~1146--1158,
  2017.

\bibitem{Zhang_diff_sizes_2015}
J.~Zhang, X.~Lin, C.~C. Wang, and X.~Wang, ``Coded caching for files with
  distinct file sizes,'' in {\em IEEE International Symposium on Information
  Theory (ISIT)}, pp.~1686--1690, June 2015.

\bibitem{ji2014order}
M.~Ji, A.~M. Tulino, J.~Llorca, and G.~Caire, ``Order-optimal rate of caching
  and coded multicasting with random demands,'' {\em IEEE Transactions on
  Information Theory}, vol.~63, pp.~3923--3949, June 2017.

\bibitem{Qian_exact_rate}
Q.~Yu, M.~A. Maddah-Ali, and A.~S. Avestimehr, ``The exact rate-memory tradeoff
  for caching with uncoded prefetching,'' {\em IEEE Transactions on Information
  Theory}, vol.~64, no.~2, pp.~1281--1296, 2018.

\bibitem{Qian_trade_off}
Q.~Yu, M.~A. Maddah-Ali, and A.~S. Avestimehr, ``Characterizing the rate-memory
  tradeoff in cache networks within a factor of 2,'' {\em IEEE Transactions on
  Information Theory}, vol.~65, no.~1, pp.~647--663, 2019.

\bibitem{Amiri_2016}
M.~M. Amiri, Q.~Yang, and D.~G{\"u}nd{\"u}z, ``Decentralized coded caching with
  distinct cache capacities,'' in {\em 50th Asilomar Conference on Signals,
  Systems and Computers}, pp.~734--738, Nov 2016.

\bibitem{Wang_2015}
S.~Wang, W.~Li, X.~Tian, and H.~Liu, ``Fundamental limits of heterogenous
  cache,'' {\em CoRR}, vol.~abs/1504.01123, 2015.

\bibitem{karamchandani2016hierarchical}
N.~Karamchandani, U.~Niesen, M.~A. Maddah-Ali, and S.~N. Diggavi,
  ``Hierarchical coded caching,'' {\em IEEE Transactions on Information
  Theory}, vol.~62, no.~6, pp.~3212--3229, 2016.

\bibitem{poularakis2016}
K.~Poularakis and L.~Tassiulas, ``On the complexity of optimal content
  placement in hierarchical caching networks,'' {\em IEEE Transactions on
  Communications}, vol.~64, pp.~2092--2103, May 2016.

\bibitem{pedarsani2016online}
R.~Pedarsani, M.~A. Maddah-Ali, and U.~Niesen, ``Online coded caching,'' {\em
  IEEE/ACM Transactions on Networking}, vol.~24, no.~2, pp.~836--845, 2016.

\bibitem{niesen2015coded}
U.~Niesen and M.~A. Maddah-Ali, ``Coded caching for delay-sensitive content,''
  in {\em IEEE International Conference on Communications (ICC)},
  pp.~5559--5564, 2015.

\bibitem{ji2015efficient}
M.~Ji, K.~Shanmugam, G.~Vettigli, J.~Llorca, A.~M. Tulino, and G.~Caire, ``An
  efficient multiple-groupcast coded multicasting scheme for finite fractional
  caching,'' in {\em IEEE International Conference on Communications (ICC)},
  pp.~3801--3806, 2015.

\bibitem{hueristics_index_coding}
M.~A.~R. Chaudhry and A.~Sprintson, ``Efficient algorithms for index coding,''
  in {\em IEEE INFOCOM Workshops 2008}, pp.~1--4, April 2008.

\bibitem{Birk:2006}
Y.~Birk and T.~Kol, ``Coding on demand by an informed source (iscod) for
  efficient broadcast of different supplemental data to caching clients,'' {\em
  IEEE/ACM Transactions on Networking}, vol.~14, pp.~2825--2830, June 2006.

\bibitem{index_coding_2011}
Z.~Bar-Yossef, Y.~Birk, T.~S. Jayram, and T.~Kol, ``Index coding with side
  information,'' {\em IEEE Transactions on Information Theory}, vol.~57,
  pp.~1479--1494, March 2011.

\bibitem{index_coding_LP}
A.~Blasiak, R.~D. Kleinberg, and E.~Lubetzky, ``Index coding via linear
  programming,'' {\em CoRR}, vol.~abs/1004.1379, 2010.

\bibitem{thapa2017interlinked}
C.~Thapa, L.~Ong, and S.~J. Johnson, ``Interlinked cycles for index coding:
  Generalizing cycles and cliques,'' {\em IEEE Transactions on Information
  Theory}, vol.~63, no.~6, pp.~3692--3711, 2017.

\bibitem{vaddi2017capacity}
M.~B. Vaddi and B.~S. Rajan, ``Capacity of some index coding problems with
  symmetric neighboring interference,'' {\em arXiv preprint arXiv:1705.05060},
  2017.

\bibitem{hardness_network_coding}
M.~Langberg and A.~Sprintson, ``On the hardness of approximating the network
  coding capacity,'' {\em IEEE Transactions on Information Theory}, vol.~57,
  pp.~1008--1014, Feb 2011.

\bibitem{Shanmugam2016}
K.~Shanmugam, M.~Ji, A.~M. Tulino, J.~Llorca, and A.~G. Dimakis.,
  ``Finite-length analysis of caching-aided coded multicasting,'' {\em IEEE
  Transactions on Information Theory}, vol.~62, pp.~5524--5537, Oct 2016.

\bibitem{naderi_ICC_2017}
N.~Naderializadeh, M.~A. Maddah-Ali, and A.~S. Avestimehr, ``Cache-aided
  interference management in wireless cellular networks,'' in {\em IEEE
  International Conference on Communications (ICC)}, pp.~1--7, May 2017.

\bibitem{naderi_IT_2017}
N.~Naderializadeh, M.~A. Maddah-Ali, and A.~S. Avestimehr, ``Fundamental limits
  of cache-aided interference management,'' {\em IEEE Transactions on
  Information Theory}, vol.~63, pp.~3092--3107, May 2017.

\bibitem{Ramakrishnan_2015}
A.~Ramakrishnan, C.~Westphal, and A.~Markopoulou, ``An efficient delivery
  scheme for coded caching,'' in {\em 2015 27th International Teletraffic
  Congress}, pp.~46--54, Sept 2015.

\bibitem{Wan_2017}
K.~Wan, D.~Tuninetti, and P.~Piantanida, ``Novel delivery schemes for
  decentralized coded caching in the finite file size regime,'' in {\em 2017
  IEEE International Conference on Communications Workshops (ICC Workshops)},
  pp.~1183--1188, May 2017.

\bibitem{Combinatorial_Optimization_Korte}
B.~Korte and J.~Vygen, {\em Combinatorial Optimization: Theory and Algorithms}.
\newblock Springer Berlin Heidelberg, 2012.

\bibitem{grotschel2012geometric}
M.~Gr{\"o}tschel, L.~Lov{\'a}sz, and A.~Schrijver, {\em Geometric algorithms
  and combinatorial optimization}, vol.~2.
\newblock Springer Science \& Business Media, 2012.

\bibitem{gurobi}
``Gurobi solver.'' http://www.gurobi.com.

\bibitem{karp1972reducibility}
R.~M. Karp, ``Reducibility among combinatorial problems,'' in {\em Complexity
  of computer computations}, pp.~85--103, Springer, 1972.

\bibitem{chvatal1979greedy}
V.~Chvatal, ``A greedy heuristic for the set-covering problem,'' {\em
  Mathematics of operations research}, vol.~4, no.~3, pp.~233--235, 1979.

\bibitem{zuckerman2006linear}
D.~Zuckerman, ``Linear degree extractors and the inapproximability of max
  clique and chromatic number,'' in {\em Proceedings of the thirty-eighth
  annual ACM symposium on Theory of computing}, pp.~681--690, ACM, 2006.

\end{thebibliography}

\appendices
\section{Proof of Lemma \ref{optimal_delivery_solution_first}} 
 \label{proof:optimal_delivery_solution_first}
The proof is based on the following result for \black{the} minimum weighted set cover problem.
\begin{lem}
\label{co:optimal_delivery_solution_first}
Let $d$ be the size of the largest set $\mathcal{P} \in \mathscr{P}$. Then, Algorithm \ref{algorithm_weighted_Set_cover} provides a $(1+ \log d)$-approximation to the minimum weighted set cover problem where $\mathcal{W}$ is the set of elements that needs to be covered, $\mathscr{P}$ is the set of subsets of $\mathcal{W}$, and for each $\mathcal{P} \in \mathscr{P}$, $\bc{\mathcal{P}}$ is the weight of subset $\mathcal{P}$ of elements \cite{chvatal1979greedy}\footnote{Note that this algorithm is exactly the one proposed in \cite{chvatal1979greedy} which has been intentionally described as in Algorithm \ref{algorithm_weighted_Set_cover} for the sake of comparison with Algorithm \ref{first_delivery}.}.
\end{lem}
Let $\mathcal{T}_1^{(1)} = \{1, 2,\ldots,\lambda_K \}$, $\mathcal{P}^{(s)}$ be the selected packet at the $s$-th iteration, and $\mathcal{P}_k^{(s)}$ be packet $\mathcal{P}_k$ at the beginning of $s$-th iteration of Algorithm \ref{algorithm_weighted_Set_cover}. Then, for $iter = 2,\ldots$, define 
\begin{align}
\label{set_T_1}
&\mathcal{T}_1^{(iter)} = \notag \\
& \{k:  k \in \mathcal{T}_1^{(1)},  \mathcal{P}_k^{(1)} \cap \mathcal{P}^{(s)} = \emptyset  \text{ for all } 1 \leq s \leq iter -1\}, \\
&\mathcal{T}_2^{(iter)} = \mathcal{T}_1^{(1)} \setminus \mathcal{T}_1^{(iter)}.
\label{set_T_2}
\end{align}
\begin{algorithm}
\caption{Weighted-Set-Cover-Solver \cite{chvatal1979greedy}}
\label{algorithm_weighted_Set_cover}
\hspace*{\algorithmicindent} {\small \textbf{Input:} Set of subfiles $\mathcal{W}$ and set of all feasible packets $\mathscr{P}$ labeled as $\{\mathcal{P}_1^{(1)}, \mathcal{P}_2^{(1)}, \ldots, \mathcal{P}_{\lambda_K}^{(1)} \}$ where $\lambda_K = |\mathscr{P}|$.}
\begin{algorithmic}[1]
\State $\mathscr{C} = \emptyset$
\State $\mathcal{E} = \mathcal{W}$
\State $iter = 1$
\While{$\mathcal{E} \neq \emptyset$}
\State $j = \argmax_{k=1,\ldots, \lambda_K} \dfrac{|\mathcal{P}_k^{(iter)}|}{\bc{\mathcal{P}_k^{(1)}}}$ 
\State $\mathcal{P}^{(iter)} = \mathcal{P}_j^{(1)}$
\State $\mathscr{C} = \mathscr{C} \cup \{\mathcal{P}^{(iter)}\}$
\State $\mathcal{E} = \mathcal{E} \setminus \mathcal{P}^{(iter)}$
\For{$k=1,\ldots, \lambda_K$}
\State $\mathcal{P}_k^{(iter+1)} = \mathcal{P}_k^{(iter)} \setminus \mathcal{P}^{(iter)}$
\EndFor
\State $iter = iter +1$
\EndWhile
\end{algorithmic}
\hspace*{\algorithmicindent} {\small \textbf{Output:} Set of packets $\mathscr{C}$.}
\end{algorithm} 
Note that $\mathcal{T}_1^{(iter)}$ denotes the set of indices of packets that do not intersect with any selected packet before iteration $iter$, and $\mathcal{T}_2^{(iter)}$ is its complement. Further, define 
\begin{align}
\mathscr{S}_1^{(iter)} = \{\mathcal{P}_k^{(iter)}: k \in  \mathcal{T}_1^{(iter)} \}.
\label{claim:closure}
\end{align}
Now, we state some preliminary results and then, use them to prove Lemma \ref{optimal_delivery_solution_first}. 
\begin{claim}
\label{claim:feasibility}
For any $k \in \mathcal{T}_1^{(1)} = \{1, 2,\ldots,\lambda_K \}$ and for any iteration $iter$, $\mathcal{P}_k^{(iter)}$ is either an empty packet or a feasible packet.
\end{claim}
\begin{proof}
To show this, note that $\mathcal{P}_k^{(1)}$ is a feasible packet (it belongs to $\mathscr{P}$). Now assume that $\mathcal{P}_k^{(iter-1)}$ is either an empty packet or a feasible packet. Note that $\mathcal{P}^{(s)}$ is a feasible packet for any iteration $s$ because it is equal to $\mathcal{P}_j^{(1)}$ for some $j \in \mathcal{T}_1^{(1)}$ and we know that all $\mathcal{P}_j^{(1)}$ are feasible. If $\mathcal{P}_k^{(iter-1)}$ is an empty packet, it is clear that $\mathcal{P}_k^{(iter)}$ is also an empty packet. Further, if $\mathcal{P}_k^{(iter-1)}$ is a feasible packet, since $\mathcal{P}^{(iter-1)}$ is also a feasible packet, 
$\mathcal{P}_k^{(iter)} = \mathcal{P}_k^{(iter-1)} \setminus \mathcal{P}^{(iter-1)}$ is clearly either an empty packet or a feasible packet.
\end{proof}

\begin{claim}
\label{claim:no_intersection}
For any $k \in \mathcal{T}_1^{(iter)}$, $\mathcal{P}_k^{(iter)}$ does not intersect with any selected packet before $iter$-th iteration, that is, $\mathcal{P}_k^{(iter)} \cap \mathcal{P}^{(s)} = \emptyset, \quad \forall \quad  1\leq s \leq iter -1$.
\end{claim}
\begin{proof}
This proof \black{results} from the update equation for $\mathcal{P}_k^{(iter)}$ given in 
Algorithm \ref{algorithm_weighted_Set_cover}.
\end{proof}

\begin{claim}
\label{claim:set_T_1_prop}
For any $k \in \mathcal{T}_1^{(iter)}$, $\mathcal{P}_k^{(1)} = \mathcal{P}_k^{(2)} = \ldots =  \mathcal{P}_k^{(iter)}$,
that is, packet $\mathcal{P}_k^{(1)}$ does not change over the first $iter-1$ iterations. 
\end{claim}
\begin{proof}
This proof \black{results} from the definition of set $\mathcal{T}_1^{(iter)}$ in \eqref{set_T_1}. Since $k \in \mathcal{T}_1^{(iter)}$, $\mathcal{P}_k^{(1)}$ does not intersect with any $\mathcal{P}^{(s)}$ for $s \leq iter-1$. Then, from the update equation for $\mathcal{P}_k^{(iter)}$ given in Algorithm \ref{algorithm_weighted_Set_cover}, the correctness of Claim \ref{claim:set_T_1_prop} \black{can be} concluded.
\end{proof}
\begin{claim}
\label{prop:sizes}
Let $\mathcal{P}_1$ and $\mathcal{P}_2$ be feasible packets and let $\mathcal{P}_3 = \mathcal{P}_1 \setminus \mathcal{P}_2$. Then, $\bc{\mathcal{P}_3} \leq \bc{\mathcal{P}_1}$.
\end{claim}
\begin{proof}
The proof \black{results} from the definition $\bc{.}$ in \eqref{bc_definition} and the fact that $\mathcal{P}_3 \subseteq \mathcal{P}_1$.
\end{proof}
Now, we show that $\mathcal{P}^{(iter)}$ selected at the $iter$-th iteration of Algorithm \ref{algorithm_weighted_Set_cover} is the same as $\mathcal{P}$ chosen at the $iter$-th iteration of  Algorithm \ref{first_delivery}. To this end, we show that 
\begin{align}
\max_{k \in \mathcal{T}_1^{(1)}} \dfrac{|\mathcal{P}_k^{(iter)}|}{\bc{\mathcal{P}_k^{(1)}}} 
= \max_{k \in \mathcal{T}_1^{(iter)}} \dfrac{|\mathcal{P}_k^{(iter)}|}{\bc{\mathcal{P}_k^{(1)}}},
\label{claim:proof_1}
\end{align}
that is, in order to find maximizer index $k \in \mathcal{T}_1^{(1)}$, we only need to focus on $k$'s belonging to set $\mathcal{T}_1^{(iter)}$. To show this, choose any \black{$k \in \mathcal{T}_1^{(iter)}$}. Then, from Claim \ref{claim:feasibility}, $\mathcal{P}_k^{(iter)}$ is either an empty packet or a feasible packet. 
\begin{itemize}
\item If $\mathcal{P}_k^{(iter)}$ is an empty packet, we have $|\mathcal{P}_k^{(iter)}| = 0$ and hence index $k$ is never chosen as long as there exists some packet $\mathcal{P}_{k'}^{(iter)}$ with $|\mathcal{P}_{k'}^{(iter)}| > 0$. Therefore, we can ignore $k$ in the optimization of \eqref{claim:proof_1}. Further, if $|\mathcal{P}_{k'}^{(iter)}| = 0$ for all $k' \in \mathcal{T}_1^{(1)}$, it can be easily observed that $\mathcal{E}=\emptyset$ at the beginning $iter$-th iteration, and hence, we do not need to solve \eqref{claim:proof_1}. 

\item If $\mathcal{P}_k^{(iter)}$ is a feasible packet, then there should be some $l_k \in \mathcal{T}_1^{(1)}$ such that $\mathcal{P}_{l_k}^{(1)} = \mathcal{P}_k^{(iter)}
$ because $\mathcal{T}_1^{(1)}$ contains the indices of all feasible packets. Then from Claim \ref{claim:no_intersection}, $\mathcal{P}_{l_k}^{(1)}$ does not intersect with any selected packet before $iter$-th iteration and hence, by definition of $\mathcal{T}_1^{(iter)}$ in \eqref{set_T_1}, we have $l_k \in \mathcal{T}_1^{(iter)}$. Then, from Claim \ref{claim:set_T_1_prop}, we have 
\begin{align}
\label{equality_set_T_1_proof}
\mathcal{P}_{l_k}^{(1)} = \mathcal{P}_{l_k}^{(2)} = \ldots =  \mathcal{P}_{l_k}^{(iter)}.
\end{align}
Furthermore, from Claim \ref{prop:sizes}, we have $\bc{\mathcal{P}_k^{(iter)}} \leq \bc{\mathcal{P}_k^{(1)}}$ and since $\mathcal{P}_{l_k}^{(1)} = \mathcal{P}_k^{(iter)}
$, we have $\bc{\mathcal{P}_{l_k}^{(1)}} \leq \bc{\mathcal{P}_k^{(iter)}}$. This means that 
\begin{align}
\dfrac{|\mathcal{P}_{l_k}^{(1)}|}{\bc{\mathcal{P}_{l_k}^{(1)}}} \geq \dfrac{|\mathcal{P}_{k}^{(iter)}|}{\bc{\mathcal{P}_{k}^{(1)}}} \Longrightarrow 
\dfrac{|\mathcal{P}_{l_k}^{(iter)}|}{\bc{\mathcal{P}_{l_k}^{(1)}}} \geq \dfrac{|\mathcal{P}_{k}^{(iter)}|}{\bc{\mathcal{P}_{k}^{(1)}}}
\label{claim:proof2}
\end{align}
where \eqref{claim:proof2} is correct because of \eqref{equality_set_T_1_proof}. Therefore, we can ignore $k$ in the optimization of \eqref{claim:proof_1} in favor of $l_k$.
\end{itemize}
Consequently, in order to find the $k$ maximizing $\frac{|\mathcal{P}_k^{(iter)}|}{\bc{\mathcal{P}_k^{(1)}}}$, we only need to focus on set $\mathcal{T}_1^{(iter)}$ and hence, \eqref{claim:proof_1} is correct. Further, note that from \eqref{equality_set_T_1_proof} and \eqref{claim:proof_1}, we have 
 \begin{align}
\max_{k \in \mathcal{T}_1^{(1)}} \dfrac{|\mathcal{P}_k^{(iter)}|}{\bc{\mathcal{P}_k^{(1)}}} 
&= \max_{k \in \mathcal{T}_1^{(iter)}} \dfrac{|\mathcal{P}_k^{(iter)}|}{\bc{\mathcal{P}_k^{(1)}}}
= \max_{k \in \mathcal{T}_1^{(iter)}} \dfrac{|\mathcal{P}_k^{(1)}|}{\bc{\mathcal{P}_k^{(1)}}} \notag \\
&
= \max_{\mathcal{P} \in \mathscr{S}_1^{(iter)}} \dfrac{|\mathcal{P}|}{\bc{\mathcal{P}}},
\label{claim:proof_3}
\end{align}
where the last equality follows from the definition of $\mathscr{S}_1^{(iter)}$ in  \eqref{claim:closure}. Now, from the definition of function $\textsc{PBO}$ and lines 5 and 6 of Algorithm \ref{algorithm_weighted_Set_cover}, \eqref{claim:proof_3} means that $\mathcal{P}^{(iter)} = \textsc{PBO}(\mathscr{S}^{(iter)}_1)$. Note that $\mathscr{S}^{(iter)}_1$ is exactly \black{the} set $\mathscr{S}$ at the $iter$-th iteration of Algorithm \ref{first_delivery}. Therefore, $\mathcal{P}^{(iter)}$ and $\mathcal{P}$ are the same at the $iter$-th iteration of Algorithm \ref{algorithm_weighted_Set_cover} and Algorithm \ref{first_delivery}. Hence, these two algorithms are equivalent.
Since Algorithm \ref{algorithm_weighted_Set_cover} provides a $(1+ \log d)$-approximation to minimum weighted set cover problem, and further this problem and Problems  \ref{optimal_delivery_N} are equivalent, Algorithm \ref{first_delivery} provides a $(1+ \log d)$-approximation to Problems  \ref{optimal_delivery_N}. Furthermore, note that in Problem \ref{optimal_delivery_N}, we have $|\mathcal{P}|\leq K$ for all $\mathcal{P} \in \mathscr{P}$. Therefore, Algorithm \ref{algorithm_weighted_Set_cover} provides a $(1+ \log K)$-approximation to Problem \ref{optimal_delivery_N}.

\section{Proof of Theorem \ref{equivalent_function}}
\label{proof:optimal_delivery_solution_second}
Let $\mathcal{E}$ be any subset of $\mathcal{W}$ and let $\mathscr{S}$ denote the set of all possible feasible packets that can be generated from $\mathcal{E}$. We define $\mathcal{M} \subseteq [K]$ as the indices of users for which there exists a subfile in $\mathcal{E}$ that needs to be sent. More specifically, 
\begin{align}
\label{M_set}
\mathcal{M} = \{k: k \in [K], \exists W_{k, \mathcal{A}} \in \mathcal{E} \text{ for some } \mathcal{A} \subseteq [K] \setminus \{k\} \}.
\end{align}
Let $\mathscr{T}$ be set of all non-empty subsets of $\mathcal{M}$, that is, $\mathscr{T} = \{\mathcal{T} \subseteq \mathcal{M}: |\mathcal{T}| \neq 0 \}$. Then, we define for $\mathcal{T} \in \mathscr{T}$,
\begin{align}
\label{space_decomposition}
\mathscr{Q}_{\mathcal{T}} = \{\mathcal{P}_{\mathcal{A}_{1:K}}:  \mathcal{P}_{\mathcal{A}_{1:K}} \in \mathscr{S}, \forall j \in \mathcal{T} \hspace{1mm} \mathcal{A}_j \neq \aleph, \forall j \not \in \mathcal{T} \hspace{1mm} \mathcal{A}_j = \aleph
\}.
\end{align}

\begin{claim}
\label{disjoint_decomposition}
Let $\mathcal{E}$ be any subset of $\mathcal{W}$ and let $\mathscr{S}$ denote the set of all possible feasible packets that can be generated from $\mathcal{E}$. Define $\mathcal{M}$ using \eqref{M_set} and $\mathscr{T} = \{\mathcal{T} \subseteq \mathcal{M}: |\mathcal{T}| \neq 0 \}$. Furthermore, for any $\mathcal{T} \in \mathscr{T}$, define $ \mathscr{Q}_{\mathcal{T}}$ using \eqref{space_decomposition}. Then the sets $ \mathscr{Q}_{\mathcal{T}}$, $\mathcal{T} \in \mathscr{T}$, are disjoint and $\cup_{\mathcal{T} \in \mathscr{T}}  \mathscr{Q}_{\mathcal{T}} = \mathscr{S}$.
\end{claim}
\begin{proof}
The proof is easily concluded from the definition of set $\mathscr{Q}_{\mathcal{T}}$ in \eqref{space_decomposition}.
\end{proof}

For any set $\mathcal{T} \in \mathscr{T}$ and any $j \in \mathcal{T}$, we define 
\begin{align}
\label{Gamma_T}
\mathcal{L}_{j,\mathcal{T}} = \{W_{j, \mathcal{A}}: W_{j, \mathcal{A}} \in \mathcal{E}, \mathcal{T} \setminus \{j\}  \subseteq \mathcal{A} \}.
\end{align}
Further, we define $\mathcal{L}_{j,\mathcal{T}} = \{W_*\}$ with $\bc{W_*} = \infty$ whenever $\mathcal{L}_{j,\mathcal{T}}$ is empty from definition above.
Note that $\mathcal{L}_{j,\mathcal{T}}$ includes subfiles in $\mathcal{E}$ needed to be sent to user $j$ which are available in the cache of all users of set $\mathcal{T} \setminus \{j\}$. If $\mathcal{T} \setminus \{j\} = \varnothing$, then $\mathcal{L}_{j,\mathcal{T}}$ includes all subfiles in $\mathcal{E}$ requested by user $j$. 

\begin{claim}
\label{equivalent_selection}
Given $\mathcal{T} \in \mathscr{T}$, let $\mathscr{Q}_{\mathcal{T}} $ be defined as in \eqref{space_decomposition} and $\mathcal{L}_{j,\mathcal{T}}$ for all $j \in \mathcal{T}$ be as defined in \eqref{Gamma_T}. 
\begin{enumerate}
\item If $\mathcal{P}_{\mathcal{A}_{1:K}}$ is a feasible packet in $\mathscr{Q}_{\mathcal{T}} $, then for each $j \in \mathcal{T}$, there is exactly one subfile $W_{j, \mathcal{A}_j} \in\mathcal{P}_{\mathcal{A}_{1:K}}$ that $W_{j, \mathcal{A}_j} \in \mathcal{L}_{j,\mathcal{T}}$.
\item If $\mathcal{P}_{\mathcal{A}_{1:K}}$ is a collection of subfiles $W_{j, \mathcal{A}_j} $, \black{one for each} $j \in \mathcal{T}$ where subfile $W_{j, \mathcal{A}_j}  \in \mathcal{L}_{j,\mathcal{T}}$, then, $\mathcal{P}_{\mathcal{A}_{1:K}}$ is a feasible packet in $\mathscr{Q}_{\mathcal{T}} $.
\end{enumerate}
\end{claim}

\begin{proof}
We first prove part 1. According to \eqref{space_decomposition}, for any packet $\mathcal{P}_{\mathcal{A}_{1:K}} \in \mathscr{Q}_{\mathcal{T}}$, we have $\mathcal{A}_j \neq \aleph$ for all $j \in \mathcal{T}$ and  $\mathcal{A}_j = \aleph$ for all $j \not \in \mathcal{T}$. For each $j \in \mathcal{T}$, since $\mathcal{A}_j \neq \aleph$, there should be a subfile $W_{j,\mathcal{A}_j} \in \mathcal{P}_{\mathcal{A}_{1:K}}$. Since $\mathcal{P}_{\mathcal{A}_{1:K}} \in \mathscr{Q}_{\mathcal{T}}$, we have $W_{j,\mathcal{A}_j} \in \mathcal{E}$. Furthermore, since $\mathcal{P}_{\mathcal{A}_{1:K}}$ is a feasible packet, we should have
$\mathcal{T} \setminus \{j\} \subseteq \mathcal{A}_j$ for all $j \in \mathcal{T}$. Hence, $W_{j,\mathcal{A}_j} \in \mathcal{L}_{j,\mathcal{T}} $ for all $j \in \mathcal{T}$. 

To prove part 2, let $\mathcal{P}_{\mathcal{A}_{1:K}} = \{W_{j,\mathcal{A}_j}: j \in \mathcal{T} \}$. Since for each $j \in \mathcal{T}$, $W_{j,\mathcal{A}_j} \in \mathcal{L}_{j,\mathcal{T}} $, by definition of $\mathcal{L}_{j,\mathcal{T}} $ in \eqref{Gamma_T}, we know that $\mathcal{T} \setminus \{j\} \subseteq \mathcal{A}_j$. This is true for all $j \in \mathcal{T}$, hence $\mathcal{P}_{\mathcal{A}_{1:K}} = \{W_{j, \mathcal{A}_j}: j \in \mathcal{T}  \}$ satisfies the definition of feasible packets. Therefore, $\mathcal{P}_{\mathcal{A}_{1:K}}$ is a feasible packet and since $\mathcal{A}_j \neq \aleph$ for all $j \in \mathcal{T}$ and  $\mathcal{A}_j = \aleph$ for all $j \not \in \mathcal{T}$, $\mathcal{P}_{\mathcal{A}_{1:K}} \in \mathscr{Q}_{\mathcal{T}}$.
\end{proof}

Claim \ref{equivalent_selection} suggests that finding packet $\mathcal{P}_{\mathcal{A}_{1:K}} \in \mathscr{Q}_{\mathcal{T}}$ with the minimum size is equivalent to finding a collection of subfiles $V_{j, \mathcal{T}}$, $j \in \mathcal{T}$ where $V_{j, \mathcal{T}}$ has the minimum size among the subfiles of set $\mathcal{L}_{j,\mathcal{T}}$. The following Lemma states this result.
\begin{claim}
\label{decompisition_optimization}
Let $\mathcal{E}$ be any subset of $\mathcal{W}$ and let $\mathcal{T} \in \mathscr{T}$. Define $\mathscr{Q}_{\mathcal{T}}$ from \eqref{space_decomposition} and for all $j \in \mathcal{T}$, define $\mathcal{L}_{j,\mathcal{T}}$ from \eqref{Gamma_T}.
Then, if $V_{j,\mathcal{T}} = \argmin_{W \in \mathcal{L}_{j,\mathcal{T}}} \bc{W}$ \black{and $\mathcal{R}_{\mathcal{T}} = \{V_{j, \mathcal{T}}: j \in \mathcal{T}\}$},we have
$\min_{\mathcal{P} \in \mathscr{Q}_{\mathcal{T}}} \bc{\mathcal{P}} = \black{\bc{\mathcal{R}_{\mathcal{T}}}}$
where we define $\mathcal{L}_{j,\mathcal{T}} = \{W_*\}$ with $\bc{W_*} = \infty$ whenever $\mathcal{L}_{j,\mathcal{T}} $ is empty by definition of \eqref{Gamma_T}.
\end{claim}
\begin{proof}
{\color{black} 
From part 1 of Claim \ref{equivalent_selection}, we know that each feasible packet $\mathcal{P} \in \mathscr{Q}_{\mathcal{T}}$ includes exactly one subfile $W_{j, \mathcal{A}_j} \in \mathcal{L}_{j,\mathcal{T}}$ for each $j \in \mathcal{T}$. Consider a feasible packet $\mathcal{P} \in \mathscr{Q}_{\mathcal{T}}$, then we can write $\bc{\mathcal{P}}$ as $\bc{\mathcal{P}} = \max_{j \in \mathcal{T}} \bc{W_{j,A_j}}$.
Since $V_{j,\mathcal{T}}$ is the subfile with the smallest size in $\mathcal{L}_{j,\mathcal{T}}$, we have $\bc{W_{j,A_j}} \geq \bc{V_{j,\mathcal{T}}}$ for each $j \in \mathcal{T}$. 
Consequently, $\bc{\mathcal{P}} \geq \max_{j \in \mathcal{T}} \bc{V_{j,\mathcal{T}}}$ for any feasible packet $\mathcal{P} \in \mathscr{Q}_{\mathcal{T}}$, which immediately implies that $\min_{\mathcal{P} \in \mathscr{Q}_{\mathcal{T}}} \bc{\mathcal{P}} \geq \max_{j \in \mathcal{T}} \bc{V_{j,\mathcal{T}}}$.
Now, from part 2 of Claim \ref{equivalent_selection}, we know that $\mathcal{R}_{\mathcal{T}} = \{V_{j, \mathcal{T}}: j \in \mathcal{T}\}$ is indeed a feasible packet. This means that $\bc{\mathcal{R}_{\mathcal{T}}} \geq \min_{\mathcal{P} \in \mathscr{Q}_{\mathcal{T}}} \bc{\mathcal{P}}$.
Furthermore, $\max_{j \in \mathcal{T}} \bc{V_{j,\mathcal{T}}} = \bc{\mathcal{R}_{\mathcal{T}}}$ from the definition of the size of a feasible packet. Then, we have
$\bc{\mathcal{R}_{\mathcal{T}}} \geq \min_{\mathcal{P} \in \mathscr{Q}_{\mathcal{T}}} \bc{\mathcal{P}} \geq \max_{j \in \mathcal{T}} \bc{V_{j,\mathcal{T}}} = \bc{\mathcal{R}_{\mathcal{T}}}$. Hence, the above inequalities must hold with equality, and the proof of the claim is complete.
}

\end{proof}
Now we use the above results to prove Theorem \ref{equivalent_function} as follows,
\begin{align}
&
 \textsc{PBO}(\mathscr{S}) =\max_{\mathcal{P}_{\mathcal{A}_{1:K}} \in \mathscr{S}} \dfrac{|\mathcal{P}_{\mathcal{A}_{1:K}}|}{\bc{\mathcal{P}_{\mathcal{A}_{1:K}}}} 
= \max_{\mathcal{T} \in \mathscr{T}} \max_{\mathcal{P}_{\mathcal{A}_{1:K}} \in \mathscr{Q}_{\mathcal{T}}} \dfrac{|\mathcal{P}_{\mathcal{A}_{1:K}}|}{\bc{\mathcal{P}_{\mathcal{A}_{1:K}}}} 
\notag \\
&
= \max_{\mathcal{T} \in \mathscr{T}} \max_{\mathcal{P}_{\mathcal{A}_{1:K}} \in \mathscr{Q}_{\mathcal{T}}} \dfrac{|\mathcal{T}|}{\bc{\mathcal{P}_{\mathcal{A}_{1:K}}}}  
= \max_{\mathcal{T} \in \mathscr{T}} \dfrac{|\mathcal{T}|}{\min_{\mathcal{P}_{\mathcal{A}_{1:K}} \in \mathscr{Q}_{\mathcal{T}}} \bc{\mathcal{P}_{\mathcal{A}_{1:K}}}} 
\notag \\
&
= \max_{\mathcal{T} \in \mathscr{T}} \dfrac{|\mathcal{T}|}{\bc{\mathcal{R}_{\mathcal{T}}}} =  \textsc{SBO}(\mathcal{E}),
\end{align}
where the first equality is true by definition of function $ \textsc{PBO}$, the second equality is true because of Claim \ref{disjoint_decomposition}, the third equality is true because for all $\mathcal{P}_{\mathcal{A}_{1:K}} \in \mathscr{Q}_{\mathcal{T}}$, $|\mathcal{P}_{\mathcal{A}_{1:K}}| = |\mathcal{T}|$, the fifth equality is true by Claim \ref{decompisition_optimization}, and the last equality is true by definition of function $ \textsc{SBO}$.

 \section{Proof of Theorem \ref{optimal_delivery_solution_second}}
\label{proof:optimal_delivery} 
First note that lines 8-10 of Algorithm \ref{first_delivery} guarantee that  $\mathscr{S}$ at end of each while loop is the set of all possible feasible packets that can be generated from $\mathcal{E}$ of line 7. Therefore from Theorem \ref{equivalent_function}, the size of the packet obtained from function \textsc{PBO} is the same as the size of the packet obtained from function \textsc{SBO}. Hence, we can replace function \textsc{PBO} with function \textsc{SBO} and this not change the optimality order of Algorithm \ref{first_delivery}. In other words, Algorithm \ref{my_delivery} provides a $(1 + \log K)$-approximation to Problem \ref{optimal_delivery_N}. 
\section{Proof of Theorem \ref{thm:perfromance_SID}}
\label{proof:thm:perfromance_SID}
We first prove part 1 of Theorem \ref{thm:perfromance_SID} by providing an instance of Problem \ref{optimal_delivery_N} on which the approximation ratio of the uncoded delivery is only $K$. To see this, consider an instance of Problem \ref{optimal_delivery_N} with $K$ users where $\mathcal{W} = \{W_{k,[K] \setminus \{k\}}: k \in [K] \}$ and $\bc{W} = B$ bits for all $W \in \mathcal{W}$. Fig. \ref{p1} indicates the graph for this instance of Problem \ref{optimal_delivery_N}.
For this instance, the total number of bits sent by \black{the} conventional uncoded delivery is $KB$ while the optimal solution to Problem \ref{optimal_delivery_N} sends packet $\mathcal{\tilde P} =\{W_{k,[K] \setminus \{k\}}: k \in [K] \}$ with the total number of bits equal to $B$. Therefore, the ratio between the number of bits sent by the conventional uncoded delivery to the number of bits sent by the optimal solution to Problem \ref{optimal_delivery_N} is $K$.

\begin{figure*}[!t]
  \begin{center}
  \resizebox{0.65\textwidth}{!}{%
\begin{tikzpicture}[minimum size=20mm, node distance=1cm, every loop/.style={},
                    thick,state/.style={circle, inner sep=0pt, draw,font=\sffamily\small\bfseries}]
\node[state] (A_1) at (3,0) {$W_{1, [K] \setminus \{1\}}$};

\node[state] (A_2) at (0,0) {$W_{2, [K-1] \setminus \{2\}}$};

\node[state] (A_3) at (-4,0) {$W_{l^*-1, [l^*-2] \cup \{l^*\}}$};

\node[state] (A_l) at (-7,0) {$W_{l^*, [l^*-1]}$};

\node[state] (A_l1) at (6,0) {$W_{l^*+1,[K] \setminus \{l^*+1\}}$};
\node[state] (A_N1) at (11,0) {$W_{K-1, [K] \setminus \{K-1\}}$};
\node[state] (A_N) at (14,0) {$W_{K, [K] \setminus \{K\}}$};

\draw[every node/.style={font=\sffamily\small}]
(A_1)  edge node {}   (A_2)
(A_N1)  edge node {}   (A_N)

 (A_1) edge [bend right = 30] node {}   (A_3)
 (A_1) edge [bend right = 40] node {}   (A_l)
 (A_1) edge node {}   (A_l1)
  (A_1) edge [bend left = 35] node {}   (A_N1)
   (A_1) edge [bend left = 37] node {}   (A_N)
 (A_2) edge [bend right = 40] node {}   (A_l)
  (A_3) edge node {}   (A_l)
   (A_l1) edge [bend left = 30] node {}   (A_N1)
      (A_l1) edge [bend left = 32] node {}   (A_N);
 
 \draw[line width = 1.0mm, dash pattern=on .05mm off 2mm,
                                         line cap=round] (8,0) -- (8.6,0);                                                               
 \draw[line width = 1.0mm, dash pattern=on .05mm off 2mm,
                                         line cap=round] (-1.8,0) -- (-1.3,0);                                                                                                       

\end{tikzpicture}
}
\end{center}
\caption{{\footnotesize Instance of Problem \ref{optimal_delivery_N} for Part 2 of Theorem \ref{thm:perfromance_SID}}}
\label{p2}
\end{figure*}

\begin{figure}
\centering
  \begin{center}
   \resizebox{0.4\textwidth}{!}{%
\begin{tikzpicture}[minimum size=20mm, node distance=1cm, every loop/.style={},
                    thick,state/.style={circle, inner sep=0pt, draw,font=\sffamily\small\bfseries}]
\node[state] (A_1) at (-7,0) {$W_{1,[K] \setminus \{1\}}$};

\node[state] (A_2) at (-4,0) {$W_{2,[K] \setminus \{2\}}$};

\node[state] (A_3) at (-1,0) {$W_{3,[K] \setminus \{3\}}$};

\node[state] (A_N1) at (3,0) {$W_{K-1,[K] \setminus \{K-1\}}$};
\node[state] (A_N) at (6,0) {$W_{K,[K] \setminus \{K\}}$};

\draw[every node/.style={font=\sffamily\small}]
(A_1)  edge node {}   (A_2)
(A_2)  edge node {}   (A_3)
(A_N1)  edge node {}   (A_N)

 (A_1) edge [bend left = 40] node {}   (A_3)
 (A_1) edge [bend left = 42] node {}   (A_N1)
 (A_1) edge [bend left = 44] node {}   (A_N)
 (A_2) edge [bend left = 42] node {}   (A_N1)
 (A_2) edge [bend left = 44] node {}   (A_N)
 (A_3) edge [bend left = 40] node {}   (A_N1)
 (A_3) edge [bend left = 42] node {}   (A_N);
 
 \draw[line width = 1.0mm, dash pattern=on .05mm off 2mm,
                                         line cap=round] (0.5,0) -- (1,0);                                                               

\end{tikzpicture}
}
\end{center}
\caption{{\footnotesize Instance of Problem \ref{optimal_delivery_N} for Part 1 of Theorem \ref{thm:perfromance_SID}}}
\label{p1}
\end{figure}
\begin{figure}
  \begin{center}
   \resizebox{0.4\textwidth}{!}{%
\begin{tikzpicture}[minimum size=20mm, node distance=1cm, every loop/.style={},
                    thick,state/.style={circle, inner sep=0pt, draw,font=\sffamily\small\bfseries}]
\node[state] (A_1) at (-7,0) {$W_{1, [K] \setminus \{1\}}$};

\node[state] (A_2) at (-4,0) {$W_{2, [K] \setminus \{2\}}$};

\node[state] (A_3) at (-1,0) {$W_{3, [K] \setminus \{3\}}$};

\node[state] (A_N1) at (3,0) {$W_{K-1, [K] \setminus \{K-1\}}$};
\node[state] (A_N) at (6,0) {$W_{K, [K] \setminus \{K\}}$};

\node[state] (B_1) at (-7,-3) {$W_{1,\{2\}}$};

\node[state] (B_2) at (-4,-3) {$W_{2,\{3\}}$};

\node[state] (B_3) at (0,-3) {$W_{K-2,\{K-1\}}$};

\node[state] (B_N1) at (3,-3) {$W_{K-1,\{K\}}$};
\node[state] (B_N) at (6,-3) {$W_{K,\{1\}}$};

\draw[every node/.style={font=\sffamily\small}]
(A_1)  edge node {}   (A_2)
(A_2)  edge node {}   (A_3)
(A_N1)  edge node {}   (A_N)
(B_1)  edge node {}   (A_2)
(B_2)  edge node {}   (A_3)
(B_3)  edge node {}   (A_N1)
(B_N1)  edge node {}   (A_N)
(B_N.north)  edge node {}   (A_1.south)

 (A_1) edge [bend left = 40] node {}   (A_3)
 (A_1) edge [bend left = 42] node {}   (A_N1)
 (A_1) edge [bend left = 44] node {}   (A_N)
 (A_2) edge [bend left = 42] node {}   (A_N1)
 (A_2) edge [bend left = 44] node {}   (A_N)
 (A_3) edge [bend left = 40] node {}   (A_N1)
 (A_3) edge [bend left = 42] node {}   (A_N);
 
 \draw[line width = 1.0mm, dash pattern=on .05mm off 2mm,
                                         line cap=round] (0.5,0) -- (1,0);                                                               
                                         
\draw[line width = 1.0mm, dash pattern=on .05mm off 2mm,
                                         line cap=round] (-2.5,-4) -- (-2.0,-4);                                                                                                        

\end{tikzpicture}
}
\end{center}
\caption{{\footnotesize Instance of Problem \ref{optimal_delivery_N} for Part 3 of Theorem \ref{thm:perfromance_SID}}}
\label{p3}
\end{figure}

Next, we prove part 2 of Theorem \ref{thm:perfromance_SID} by providing an instance of Problem \ref{optimal_delivery_N} on which the approximation ratio of the greedy coded multicast (GCM) scheme is only $\lfloor \frac{K-1}{2}\rfloor$. We use $[k]$ denote the set $\{1,\ldots,k\}$. Consider an instance of Problem \ref{optimal_delivery_N} with $K$ users and $\mathcal{W}= \mathcal{W'} \cup \mathcal{W''}$ where $\mathcal{W'}= \{W_{1,[K] \setminus \{1\}}, W_{2, [K-1] \setminus \{2\}}, W_{3, [K-2] \setminus \{3\}},\ldots,
W_{l, [K-(l-1)] \setminus \{l\}}  \}$, $\mathcal{W''}= \{W_{k, [K] \setminus \{k\}}: k \in [K] \setminus [l]\}$, $\bc{W} = B$ bits for all $W \in \mathcal{W'}$, and $\bc{W} = \epsilon$ bits for all $W \in \mathcal{W''}$ such that $\epsilon \ll B$. For Problem \ref{optimal_delivery_N}, the packet $\mathcal{W'}$  is a feasible packet if $K-(l-1) \geq l$ (because we should have $[l-1] \subseteq [K-(l-1)] \setminus \{l\}$). Therefore, $\mathcal{W'}$ is feasible if $l \leq \frac{K+1}{2}$. We denote the largest value for $l$ as $l^* = \lfloor \frac{K+1}{2}\rfloor$ (note that $l$ should be an integer) and set $l = l^*$. Fig. \ref{p2} indicates the graph for this instance of Problem \ref{optimal_delivery_N}. According to GCM scheme \cite{maddah2014fundamental,maddah2015decentralized,niesen2017coded,Zhang_diff_sizes_2015}, the packets $\mathcal{P}_1 = \{W_{1,[K] \setminus \{1\}}\} \cup  \mathcal{W''}$, $\mathcal{P}_2 = \{W_{2, [K-1] \setminus \{2\}} \}$, $\mathcal{P}_3 = \{W_{3, [K-2] \setminus \{3\}} \}$, ..., and $\mathcal{P}_{l^*} = \{W_{l^*, [K-(l^*-1)] \setminus \{l^*\}} \}$ are sent with the total number of bits equal to $Bl^*$. However, since $\mathcal{W'}$ is feasible (see Fig. \ref{p2}), the optimal solution to Problem \ref{optimal_delivery_N} sends packets $\mathcal{\tilde P}_1 = \mathcal{W'}$ and $\mathcal{\tilde P}_2 = \mathcal{W''}$ with the total number of bits of $B + \epsilon$. Therefore, the ratio between the number of bits sent by the GCM scheme to the number of bits sent by the optimal solution to Problem \ref{optimal_delivery_N} is $\frac{Bl^*}{B+ \epsilon} \geq l^*-1 = \lfloor \frac{K-1}{2}\rfloor$. \black{Furthermore, under our assumption that users request different files, one can easily observe that GCC, GCLC, and HgLC schemes all will simplify to the GCC algorithm. Therefore, for the example we provided above, the ratio between the number of bits sent by these schemes to the number of bits sent by the optimal solution to Problem \ref{optimal_delivery_N} is also $\frac{Bl^*}{B+ \epsilon} \geq l^*-1 = \lfloor \frac{K-1}{2}\rfloor$.}

Finally, we prove part 3 of Theorem \ref{thm:perfromance_SID} by providing an instance of Problem \ref{optimal_delivery_N} on which the approximation ratio of the graph coloring-based coded multicast (GCCM) scheme is only $K-1$.  Consider an instance of Problem \ref{optimal_delivery_N} with $K$ users and $\mathcal{W} = \mathcal{W'} \cup \mathcal{W''}$ where 
$\mathcal{W'} = \{W_{k, [K] \setminus \{k\}}: k \in [K] \}$, 
$\mathcal{W''} = \{W_{1, \{2\}}, W_{2, \{3\}}, \ldots, W_{K-1, \{K\}}, W_{K, \{1\}}\}$, $\bc{W} = B$ bits for all $W \in \mathcal{W'}$ and $\bc{W} = \epsilon$ bits for all $W \in \mathcal{W''}$ such that $\epsilon \ll B$. Fig. \ref{p3} indicates the graph for this instance of Problem \ref{optimal_delivery_N}. For this instance of Problem \ref{optimal_delivery_N}, if we assume that the GCCM scheme \cite{ji2014order} picks vertices in the following order $W_{1, \{2\}}, W_{2, \{3\}}, \ldots, W_{K-1, \{K\}}, W_{K, \{1\}}, W_{1, [K] \setminus \{1\}}, W_{2, [K] \setminus \{2\}},\ldots,$ $W_{K, [K] \setminus \{K\}}$, then this algorithm chooses packets 
$\mathcal{P}_1 = \{W_{1,\{2\}}, W_{2, [K] \setminus \{2\}}\}$, 
$\mathcal{P}_2 = \{W_{2,\{3\}}, W_{3, [K] \setminus \{3\}}  \}$, ..., 
$\mathcal{P}_{K-1} = \{W_{K-1, \{K\}}, W_{K, [K] \setminus \{K\}} \}$, and
$\mathcal{P}_K = \{W_{K,\{1\}}, W_{1, [K] \setminus \{1\}}\}$ with the total number of bits equal to $KB$. On the other hand, for Problem \ref{optimal_delivery_N}, the packets $\mathcal{\tilde P}_1 = \{W_{1,\{2\}}\}$, $\mathcal{\tilde P}_2 = \{W_{2,\{3\}}\}$, ..., $\mathcal{\tilde P}_{K-1} = \{W_{K-1,\{K\}}\}$, $\mathcal{\tilde P}_K = \{W_{K,\{1\}}\}$, and  $\mathcal{\tilde P}_{K+1} =\mathcal{W'}$ are feasible packets. Therefore, the optimal solution to Problem \ref{optimal_delivery_N} sends these packets with the total number of bits of $K\epsilon + B$. Therefore, the ratio between the number of bits sent by the GCCM scheme to the number of bits sent by the the optimal solution to Problem \ref{optimal_delivery_N} is $\frac{KB}{B+ K\epsilon} \geq K-1$. 
\section{Proof of Theorem \ref{hardness_approximation}}
\label{proof:hardness_approximation}
First, we show that the minimum clique cover problem defined below is a special case of Problem \ref{optimal_delivery_N}.
\begin{problem}[Minimum Clique Cover Problem]
\label{clique_cover_problem}
Given graph $\mathcal{G} = (\mathcal{V}, \mathcal{E})$, find a set $\mathscr{C}$ of cliques such that every vertex $V \in \mathcal{V}$ is a member of at least one clique $\mathcal{C} \in \mathscr{C}$ and $|\mathscr{C}|$ is minimum.
\end{problem}
To this end, consider any graph $\mathcal{G} = (\mathcal{V},\mathcal{E})$ with $K$ vertices, that is, $|\mathcal{V}| = K$. We want to solve the minimum clique cover problem (Problem \ref{clique_cover_problem}) for this graph. Let's label vertices of this graph with $V_1,\ldots,V_K$. For each $k \in [K]$, define $\mathcal{A}_k =\{m: m \in [K] \setminus \{k\}, (V_k,V_m) \in \mathcal{E} \}$. That is, $\mathcal{A}_k$ includes the indices of all vertices connected to vertex $V_k$. The sets $\mathcal{A}_k$ for all $k \in [K]$ can be computed in polynomial\black{-}time in $K$. Now, consider graph $\mathcal{G}_c$ with the same vertices and edges where we have relabeled each vertex $V_k$ by $W_{k,\mathcal{A}_k}$ and all vertices have the weight of 1. Now, any clique of graph $\mathcal{G}_c$ is a feasible packet in Problem \ref{optimal_delivery_N}. This means that, graph $\mathcal{G}_c$ becomes an instance of Problem \ref{optimal_delivery_N} where $\mathcal{W} = \{W_{k,\mathcal{A}_k}: k \in [K] \}$, and since $\bc{W_{k,\mathcal{A}_k}} =1$ for all $k \in [K]$, all feasible packets of $\mathscr{P}$ have the size of 1. Therefore, solving the minimum clique cover problem for graph $\mathcal{G}$ is the same as solving Problem \ref{optimal_delivery_N} for the graph $\mathcal{G}_c$. In other words, any instance of the minimum clique cover problem can be written as a special case of Problem \ref{optimal_delivery_N}. 

Now, by contradiction, assume there is a polynomial\black{-}time (in the number $K$ of users) algorithm for Problem \ref{optimal_delivery_N} with the approximation ratio of $K^{1-\varepsilon}$ for some $\varepsilon>0$. Since we have shown that the minimum clique cover problem is a special case of Problem \ref{optimal_delivery_N}, it means that this algorithm can approximate the minimum clique cover problem with the ratio of $K^{1-\varepsilon}$ for some $\varepsilon>0$. However, this is a contradiction because we know that unless P = NP, there is no polynomial\black{-}time algorithm for \black{the} minimum clique cover problem on \black{a} $K$-vertex graph with the approximation ratio of $K^{1-\varepsilon}$ for any $\varepsilon>0$ \cite{zuckerman2006linear}. 

\end{document}